\documentclass[aps, pra, superscriptaddress, reprint, floatfix, showkeys]{revtex4-2}	

\RequirePackage[l2tabu, orthodox]{nag}

\usepackage[utf8]{inputenc}
\usepackage{multirow}
\usepackage{amsfonts, amsmath, amsthm, amssymb} 
\usepackage{mathtools}
\usepackage{soul}

\usepackage{float}

\usepackage[caption=false]{subfig}

\usepackage{etoolbox}
\apptocmd{\sloppy}{\hbadness 10000\relax}{}{}

\usepackage{tikz}

\usetikzlibrary{automata, positioning}

\usepackage{standalone}

\usepackage{braket}
\usepackage{enumitem}

\usepackage{graphicx} 

\definecolor{myurlcolor}{rgb}{0,0,0.7}
\definecolor{myrefcolor}{rgb}{0.8,0,0}

\usepackage[
	breaklinks,
	pdftex,
	pdfauthor={Maciej Stankiewicz},
	pdftitle={Benchmarking weak randomness in Quantum and Natural Sources},
	pdfsubject={Quantum Information Theory},
	pdfkeywords={randomness, amplification, earthquake, quantum},
	pdfproducer={Latex with hyperref},
	pdfcreator={pdflatex},
	colorlinks=true, 
	linkcolor=myrefcolor,
	citecolor=myurlcolor,
	urlcolor=myurlcolor
]{hyperref}

\usepackage{cleveref}

\usepackage{orcidlink}

\newtheorem{lemma}{Lemma}[section]

\newtheorem{definition}{Definition}[section]

\newtheorem{axiom}{Axiom}

\newtheorem{remark}{Remark}[section]

\newtheorem{observation}{Observation}

\definecolor{ppblue}{RGB}{46,117,182}
\definecolor{ppred}{RGB}{197, 90, 17}
\definecolor{armygreen}{rgb}{0.29, 0.33, 0.13}
\definecolor{asparagus}{rgb}{0.53, 0.66, 0.42}

\definecolor{brandeisblue}{rgb}{0.0, 0.44, 1.0}
\definecolor{brightmaroon}{rgb}{0.76, 0.13, 0.28}

\def\>{\rangle}
\def\<{\langle}

\begin{document}

\title{
Benchmarking weak randomness in Quantum and Natural Sources
}

\author{Maciej Stankiewicz \orcidlink{0000-0001-8288-3860}}
\email[Correspondence: ]{maciej@stankiewicz.edu.pl}
\affiliation{Institute of Informatics, Faculty of Mathematics, Physics and Informatics, University of Gdańsk, Wita Stwosza 57, 80-308 Gdańsk, Poland}

\author{Roberto Salazar \orcidlink{0000-0003-1737-1433}}
\affiliation{Department of Communications \& Computer Engineering, Faculty of Information \& Communication Technology (ICT), University of Malta, Msida, MSD 2080, Malta}
\affiliation{Faculty of Physics, Astronomy and Applied Computer Science, Jagiellonian University, 30-348 Kraków, Poland}

\author{Mikołaj Czechlewski \orcidlink{0000-0002-4184-2779}}
\affiliation{Institute of Informatics, National Quantum Information Centre, Faculty of Mathematics, Physics and Informatics, University of Gdańsk, Wita Stwosza 57, 80-308 Gdańsk, Poland}

\author{Alejandra Mu\~{n}oz Jensen \orcidlink{0009-0002-3894-160X}}
\affiliation{EMGG TerraData Limitada, Concepción, Chile}

\author{Catalina Morales-Yá\~{n}ez \orcidlink{0000-0002-3230-1014}}
\affiliation{Departamento de Geofísica, Universidad de Concepción, Concepción, Chile}
\affiliation{Department of Civil Engineering, Universidad Católica de la Santísima Concepción, Concepción, Chile}

\author{Omer Sakarya \orcidlink{0000-0001-8667-7474}}
\affiliation{Institute of Informatics, National Quantum Information Centre, Faculty of	Mathematics, Physics and Informatics, University of Gdańsk, Wita Stwosza 57, 80-308 Gdańsk, Poland}

\author{Julio Viveros Carrasco}
\affiliation{Department of Geophysics, Universidad de Concepción, Concepción, Chile}

\author{Stephen Walborn \orcidlink{0000-0002-3346-8625}}
\affiliation{SeQure Quantum SpA, Concepcion, Chile}

\author{Gustavo Lima \orcidlink{0000-0001-7670-6032}}
\affiliation{SeQure Quantum SpA, Concepcion, Chile}

\author{Karol Horodecki \orcidlink{0000-0001-7540-4147}}
\affiliation{Institute of Informatics, National Quantum Information Centre, Faculty of Mathematics, Physics and Informatics, University of Gdańsk, Wita Stwosza 57, 80-308 Gdańsk, Poland}

\date{\today}

\begin{abstract}
Private randomness is a fundamental resource for cryptography, security proofs, and information processing. Quantum devices offer a unique advantage by amplifying weak randomness sources in regimes unattainable by classical means. A central theoretical model for such sources is the Santha–Vazirani (SV) model, yet identifying natural processes that satisfy this model remains a major challenge. Here we take three steps toward addressing this problem. First, we introduce an axiomatic framework for quantifying weak randomness, providing a unified basis for estimating an SV-type source. Second, we develop SVTest, a general-purpose software tool for estimating the SV parameter of an arbitrary data sequence. Third, we apply this framework to both engineered and natural sources. Using data from a self-certifying commercial quantum random number generator with guaranteed min-entropy as a benchmark, we validate the accuracy and limitations of our estimation method. We then analyze geophysical signals associated with seismic activity and find that, depending on the discretization, both earthquakes and local seismic noise can exhibit SV-type randomness. Our results indicate that geophysical phenomena may constitute viable sources of cryptographic randomness, establishing an unexpected connection between quantum information theory and geophysics.
\end{abstract}


\maketitle


\section{Introduction}
\label{sec:introduction}

The creation of random bits is ubiquitous across Internet applications and extends its importance to realms such as online banking, where confidentiality against potential adversaries becomes paramount \cite{Bera2017}. 
However, achieving privacy in these bits poses an elusive challenge. 
Proving the privacy of a sequence against adversaries without additional assumptions about the source generating the sequence remains an insurmountable task.
An area of considerable interest lies within the domain of pseudorandom deterministically generated sequences, extensively investigated in \cite{Viega2003} and references therein. 
Despite their popularity, these sequences reveal vulnerability upon

partial disclosure of their initial conditions, resulting in predictability and compromising their security. 
This susceptibility to attacks challenges their seamless integration, underscoring the intricate equilibrium required between security and operational efficiency to pursue robust random bit generation.
One approach to address this challenge involves leveraging physical phenomena as foundational sources, harnessing properties of these natural processes to generate a stream of inherently weak random bits. 
However, this raw randomness requires further refinement through post-processing techniques aimed at distilling private, secure randomness---a process thoroughly examined in \cite{Bera2017, BertaFawziScholz} and references therein.
Some of the above physically certified randomnesses include radioactive decay, specific astronomical data selections, or even fluctuations captured by the camera on a mobile device \cite{Bouda2009}. 
Nevertheless, for this methodology to prove effective, it necessitates at least two statistically independent sources against classical adversaries with access to specific knowledge about the source.
Indeed, M.\ Santha and U.\ Vazirani proved a no-go statement for classical post-processing methods \cite{SV}. 
They considered a family of sources---now called the Santha-Vazirani source---parameterized by $\epsilon>0$, expressing its divergence from the ideal source of uniform bits. 
They showed that extracting a more random sequence---that is, decreasing the value of $\epsilon$---is unattainable by classical methods. 
Here, the following condition constrains $n$ bits of the SV source:
\begin{equation}
\label{eq:sv_intro}
    \frac{1}{2} -\epsilon \leq P(s_i|s_{i-1}, \ldots, s_0, e) \leq \frac{1}{2} + \epsilon
\end{equation}
for all $i\in\{1, \ldots, n\}$ and $|P(s_0|e) - \frac{1}{2} | \leq \epsilon$, where $e$ represents any knowledge about $s_i$ that the adversary may possess (see \cite{SV} for further details). 

This is a striking result since the source has a straightforward structure. It is a mixture of certain permutations of an iid Bernoulli distribution biased by $\epsilon$ around $\{\frac{1}{2}, \frac{1}{2}\}$ distribution \cite{Grudka2014}.

It is also important to mention here that classical randomness amplification is possible using two independent SV sources. 
It is a problematic approach, since it is hard to find two fully uncorrelated SV sources, and even harder to prove or estimate that independence. 
Our testing approach can be used to estimate the quality of such two sources, but independence has to be guaranteed by another means.
Therefore, in the later parts of the paper we focus more on quantum randomness amplification.

On the other hand, private randomness, an early recognized quantum resource \cite{Bera2017}, arises from quantum mechanical axioms guaranteeing randomness in specific measurement outcomes. 
Renner and Colbeck's breakthrough \cite{Colbeck2012} defied the Santha-Vazirani limitation, revealing a method to overcome it utilizing statistically independent quantum devices decoupled from the weak randomness source. 
Optimizing this process produced practical protocols using only two devices \cite{Gallego2013, Brando2016}.

Two primary models underpin private randomness: the Santha-Vazirani source, previously mentioned, and the H$_{\mathrm{min}}$ source \cite{Bouda2014}. 
Our focus is on estimating the privacy level of weak random source, assuming that it conforms to the SV source. However, our developed software could be easily modified to estimate the H$_{\mathrm{min}}$ source as well. 
Notably, recent advances

have employed heuristic modeling to characterize heartbeat signals as SV sources \cite{HeartSV}.


In this article, we go beyond these results in three ways. 
Firstly, we provide a mathematical framework for estimating a parameter $\epsilon$ in a novel axiomatic way.
Secondly, we develop software to perform such an estimation on arbitrary data.
Third, we focus on a different physical source, a more efficient one: (i) a quantum random number generator, (ii) Earthquakes, and (iii) Earth vibrations (seismic noise).

Our approach further develops the technique of attributing $\epsilon$ to a sequence of bits and goes beyond the limitations of the low amount of data that the human heart can generate within a day. 
We base this on the plausible assumption that local and global Earth's vibrations are primarily unpredictable. 
We then note that techniques for privatizing such a priori unpredictable source using quantum devices are known \cite{RotemKessler}. 
We also provide open-source software that enables the estimation of the parameter $\epsilon$ for an input sequence of bits.

\subsection{Private randomness amplification---previous works}

This section presents the background for the problem of randomness amplification and the history of results in this domain compactly, based on the comprehensive review \cite{PirandolaEtAl2020}. 

Let us first observe that generating randomness using a quantum device is straightforward when we trust its inner workings.
Preparing a quantum bit (qubit) in the equal superposition $\ket{+}\coloneqq 1/{\sqrt{2}}(\ket{0} + \ket{1})$  and measuring it in the computational basis enables the generation of random outputs.

Indeed, the Born rule \cite{Born1926} dictates unbiased outcomes for these measurements: $\frac{1}{2} : \ket{0}$ and $\frac{1}{2}: \ket{1}$. 
However, trusting random number generators (RNGs) remains questionable. 
An eavesdropper might have modified the devices during manufacturing, leading them to exhibit predictable behavior advantageous to an adversary in collaboration with the manufacturer. 
The history of successful attacks on classical RNGs is well-documented, particularly following the seminal Trojan hardware attack \cite{Becker2013}. 
Based on altering the dopant level of three transistors within the RNG circuit, the attack introduces a subtle modification to the overall structure, rendering it challenging to detect while compromising the security of randomness.

One of the possible solutions employs quantum device-independent private random number generators relying on a short but secure uniformly random seed called (quantum) private randomness expansion \cite{ColbeckPHD, Pironio2010, Colbeck2011, Acn2016}.
The seed enables the selection of inputs to the device(s) while subsequent outputs undergo processing via quantum-proof extractors. 
Upon the violation of a Bell inequality \cite{Bell1964}---an assessment of the inputs and outputs of a device surpassing a predetermined threshold---the involved parties ascertain that the resulting larger output string maintains both uniformity and remains undisclosed to potential adversaries.

Furthermore, the application scope of device certification via Bell violations has expanded significantly. This approach not only validates the randomness of quantum RNGs \cite{Pironio2010, Remi2023} and secures QKD in quantum communication \cite{PirandolaEtAl2020, Primaatmaja2023} but also proves properties of a device, such as privacy of randomness based solely on the statistics of inputs and outputs of the device, facilitating device-independent quantum information processing \cite{arnon2020}. 
Moreover, it ensures the integrity of quantum networks \cite{Meyer2022,luo2023}, verifies quantum measurements \cite{Datta2021, sarkar2023}, and guarantees fidelity in quantum computing \cite{Sekatski2018}. 
Additionally, it explores fundamental physics \cite{Pfister2016} and validates quantum computational advantages \cite{bravyi2018}. 
The comprehensive nature of Bell violations certification spans numerous applications in quantum technology, ensuring robust security, reliability, and trust.

Here, it's noteworthy to recognize that the adversary's limitations could arise from quantum mechanics, defining a quantum adversary, or from the inherent impossibility of superluminal signaling, defining a non-signaling adversary. 
Researchers have explored setups for generating private randomness in both of these contexts \cite{CSW, Mironowicz2015, ramanathan2020practical, Zhao2023}.

Furthermore, extractors are functional applications \cite{Nisan96, Trev99, RAZ2002}, separating the function's output from the adversary's memory---a repository encompassing all conceivable knowledge (refer to \cite{Bera2017, PirandolaEtAl2020} for comprehensive reviews). 
However, this solution poses a critical challenge, as it necessitates a perfectly random and secure seed, a requirement lacking practical justification.

This problem, however, was resolved by the idea of quantum \textit{private randomness amplification}, on developing which we focus in this manuscript. 
According to this approach, the honest parties have access to a private weak random source (the SV or, in general, H$_{\mathrm{min}}$ source). 
The inputs to the quantum device(s) employ these seeds, and the resultant outputs, combined with other segments of the sequence from the weak random source, contribute to forming the final bit-string, which approaches a state of near-perfect security and randomness.

The above idea is due to the seminal result of R.\ Renner and R.\ Colbeck. 
They proposed the notion of private randomness amplification, showcasing the potential to achieve nearly ideal private randomness by ideally violating the chained Bell inequality \cite{BRAUNSTEIN199022} using device inputs from the $\epsilon$-SV source, assuming $\epsilon < 0.058$. 
In \cite{Grudka2014}, this range expanded to $\epsilon \approx 0.0961$. 
Furthermore, \cite{Gallego2013} illustrated a protocol utilizing a distinct Bell inequality to amplify randomness for \textit{any} $\epsilon \in (0, \frac{1}{2})$. 
However, this approach necessitates numerous non-signaling devices and lacks noise tolerance. Remarkably, \cite{Brando2016} introduced a noise-tolerant protocol achieving similar results but with a finite device count.

The protocol given in \cite{Ramanathan2016} achieved further private randomness amplification under minimal assumptions based on two non-signaling components of a device, such that only specific inputs generate outcomes partially uncorrelated from the non-signaling adversary. 
The results of \cite{Brando2016, Ramanathan2016} rely on the premise that the adversary lacked knowledge of the weak random source; they could only know a parameter of the source (i.e., ``e'' in Eq.\ (\ref{eq:sv_intro}), not specific values of the sequence. 
The previous unrealistic assumption was later replaced in \cite{RotemKessler} by the idea of the ''privatization'' of weak random source. 
Specifically, the source generates its bits individually and makes them publicly available to both the honest parties and the adversary. 
However, before their disclosure, all parties, including the adversary, were only aware of the parameter $\epsilon$.

Nowadays, a dedicated hardware is built for implementing some quantum protocols (for instance QKD protocols). Perhaps quantum randomness amplification protocols would become commercially mature in the near future. However, practical implementation of quantum randomness amplification protocols on quantum computers was presented in \cite{Foreman2023practicalrandomness}. Hence,  although quantum computers are in the early stage of adoption  in the future, running quantum randomness protocol as software on a general purpose quantum computer could be more convenient than constructing dedicated hardware for that goal.

A more realistic but hard-to-work-with model of a source of weak randomness is the min-entropy source. 
It is described by a single condition---a lower bound on the maximal probability. 
A $k$-min-entropy source satisfies  $\log_2 \max_i p_i \geq k$ where the maximum is taken over all events that realize the source. 
A practical randomness amplification using only one two-component device and a min-entropy source with two blocks having each enough min-entropy has been given recently in \cite{RR2023}. 

In contrast, the SV-source model transfers the complexity of the tests to the conditions imposed on the classical source of weak randomness. 
Such requirements urge researchers to delve into new phenomena wherein the physical context ensures them. 
Our work exploits novel physical phenomenologies along this line of research, motivating further advancements in randomness amplification protocols.

This paper is organized as follows.
First in Section \ref{sec:main} we discuss our motivations and briefly introduce our results.
Next, in Section \ref{sec:data} we describe the types of data used in our experiments.
Then, in Section \ref{sec:methods} we formally present our mathematical framework used to testing weak random sources.
Next, in Section \ref{sec:svtest} we provide the technical aspect of the SVTest software that is developed as part of this project.
Then, in Section \ref{sec:results} we show numerical results of our experiments.
Additionally, in Section \ref{sec:avarenes} we briefly discuss the resource aware aspect of our approach.
Finally, in Section \ref{sec:discussion} we conclude our manuscript with discussion and provide open questions. 

\section{Motivation and main results}
\label{sec:main}

Since neither uniformity nor privacy of a given source of randomness can be proven, it always has to be \textit{assumed} that a given source of randomness is weak random source.
However, further, one needs to attribute some parameter---a real number, which gives us an estimate of to what extent a given source is random. 
As it was stated in \cite{PirandolaEtAl2020}:\\
\textit{``Whether a string is random or not is ultimately not a property of the string itself, but on how it is generated''}.
For this reason, the way to attribute the physical parameter that reports the quality of the randomness should be the same for all bit-strings of the same length, under the assumption that it is generated by some SV source. 
The first attempt to attribute the quality parameter to the supposed-to-be SV source was made in \cite{HeartSV}, however, in a heuristic manner and modeling a post-processed human's heartbeat.

In what follows, we provide natural axioms that such attributed $\epsilon$ should satisfy, providing several other definitions of $\epsilon$ that satisfy the axioms.
We also provide software that computes their value. 
Using this software, we further compute the value of $\epsilon$ modeling the Earth vibration (local and global separately) as the $\epsilon$-SV source.
The results are promising since there exist quantum-proof extractors that could further amplify the privacy of this particular source of randomness.

Our choice of weak random source has an important feature: it is practically impossible for a manufacturer to correlate an amplifying device with the values of the bits to profound seismic events. 
This fact is essential, as otherwise attacks on devices are known, and there does not exist a complete countermeasure against such attacks \cite{Wojewodka2017}.

 It is worth emphasizing that, in addition to theoretical and numerical results, we have delivered an open-source software for estimating the privacy of a weak random source modeled as the SV source.
We denote our software as SVTest and provide its description in the latter part of the paper (see Section \ref{sec:svtest}). 

It is important to relate our test to existing tests available online. There are two similar, however, different approaches, which are worth mentioning here. The first is the so-called serial test from the seminal NIST test suite \cite{NIST}. 
The serial test is focused on checking the frequency of all overlapping bit sequences of length $h$. The three main differences in comparison to our tests are the following:
\begin{itemize}
	\item we calculate conditional values instead of frequencies, 
	\item we use the maximal absolute deviation instead of the average square deviation, and
	\item we do not use the cyclic approach at the end of the sequence.
\end{itemize}

Also, the NIST test suite was developed to check pseudo-random sequences to use them directly in classical algorithms. 
That approach demands the sequence to be almost perfectly random. 
On the other hand, since we apply the quantum randomness amplification method, it is enough that the sequence is partially random, assuming that we know the threshold $\epsilon$.

A nice and clear way of understanding the randomness and its connection both to the Borel-normality \cite{Borel-normality} and NISTS tests was introduced in \cite{Solis_2015}. Moreover, in two papers: Rojas et.\ al.\ \cite{DazHernndezRojas2017} and  Martinez et.\ al.\ \cite{Martnez2018}  are presented quite similar but different testing methods. The first one  is based on the Borel-normality criterion, while the second one on Bayesian model selection. It is worth to emphasize that in our case the test takes into account subsequences of consecutive bits that can overlap, while in the both mentioned approach no overlapping is considered.
Furthermore, they claim that the longest length of a considered subsequence should be of the order of $\log_2(\log_2(n))$, while our test takes into account subsequences of length $\log_2(n)$. 

\section{Seismic data}
\label{sec:data}

Since its formation, the planet has manifested natural movements, pulsations, and vibrations due to interactions of fluids and solids caused by heat exchange. 
In particular, seismic waves correspond to mechanical waves of acoustic energy that travel through the media and their surfaces. 
In the case of the Earth, these can be caused naturally by earthquakes, volcanic eruptions, magmatic movements, and landslides. 
They can also be generated artificially by impacts, explosions, and industrial processes \cite[e.g.,][]{Keranen-InduceSeismicity}.

Seismic waves can be classified into two large groups: surface and body waves. 
As their name suggests, surface waves propagate through the Earth's surface, while body waves propagate through the interior of the Earth. 
The principal's body waves are the P-waves (primary or principal) and the S-waves (secondary, shear). 
P-waves are faster and are the first wave to arrive, presenting a longitudinal movement of compression and expansion, and the particles of material affected by its passage move back and forth in the same direction of wave propagation, like an accordion. 
S-waves are a little slower than P-waves; they arrive in the second position as a pure wave, their movement is transverse (shear), and the affected particles move in a perpendicular direction, vertically for S\textsubscript{V}-waves and horizontally for S\textsubscript{H}-waves (for review see \cite[e.g][]{Udias2014, shearer2019, stein2009}  
The instruments that record the Earth's vibrations are called Seismometers \cite{Udias2014}.

Their composition is equivalent (nowadays, instruments are digital) to three masses held by springs, each with a degree of freedom of movement in the three spatial components: vertical, north-horizontal, and east-horizontal. 
Each sensor can record the vibration within a range of amplitude or intensity, frequency, and duration (continuous or triggered by an event), which depends on each instrument.

This instrument is designed to capture the Earth's vibrations caused by earthquakes, including body and surface waves that convey details about the event's magnitude and geometry and insights into the Earth's interior \cite{Lillie1998-mh, Udias2014}.

Additionally, some seismological stations continuously record environmental vibrations of the ground, natural or artificial, called noise, which do not correspond to specific seismic events. 
The stations are deployed all around the planet, and the seismic records are open and free for everyone who needs to use them \cite{Udias2014}. 
Each station is generally connected to multiple national and/or international seismological services, such as the SSN of Mexico (www.ssn.unam.mx), the CSN of Chile (www.csn.uchile.cl), the USGS of the USA (www.usgs.gov), or EarthScope (IRIS \cite{iris} and UNAVCO \cite{iris2}), where the records are stored and available in various formats.

\subsection{Type of the seismic data}

The seismic information obtained to prove randomness corresponds to two data types obtained from seismological stations. 
The first corresponds to waveforms from different earthquakes, and the second corresponds to noise recorded by certain stations.

\subsubsection{Waveforms}

To obtain the waveform of the earthquakes, we selected all the events from the GCMT catalog \cite{CMT1, CMT2}, from 1976–-2021, with moment magnitude Mw (a magnitude based on the amount of energy liberated by the earthquake) between Mw = 6.0 and Mw = 7.0, and with epicentral depths between 30–-1000~km, to avoid events that occurs in the earth's crust, that are generally more complex.
There were 2218 events that fulfill above criteria.
For each earthquake, we downloaded, from the EarthScope seismic service, the record of all the available stations in a range of 10 and 50 degrees of epicentral distance (stations located at between $\sim$ 1000--5000~km from the epicenter). 
It is important to mention that to increase the aleatory of the data and due to the large amount of data, we use all the available earthquake-station pairs that follow the previous criteria (including those that may have incorrect instrument responses or clipping signals).

Then, we transform the raw signal to displacement, velocity, and acceleration and filter the signal between 0.1~Hz and 200~Hz, obtaining the principal body waves and the high-period signals. 
Then, we cut the signal with dynamic windows into two sections. 
The first window had been selected from the P-wave time arrival until $15[s/^\circ]*\Delta$ after its arrival, with $\Delta$ the epicentral distance in degrees, containing principally body waves. 
The second window starts at the S-wave time arrival to $35[s/^\circ]*\Delta$ after its arrival, where the surface waves would be predominant \cite{Duputel2012}. 
Once we have cut the signal, we unite (concatenate) each one of the time windows for all the stations and earthquakes. 
Finally, we will have six files corresponding to the two different time windows for the three signal responses (displacement, velocity, and acceleration).

\subsubsection{Noise}

For noise signals, we selected stations close to populated areas to increase human noise in the records. 
The seismological station selected corresponds to Chile, Argentina, Iceland, Indonesia, Malaysia, Australia, Nepal, India, and the USA, and the time window corresponds to 24 hours, respectively. 
These instruments, corresponding to broadband stations(BH, HH), particularly have a wide range of samplers per second compared with other kinds of stations, such as the long period ones (LH); however, to uniform the data, we resample the data to 4 samples per second.

As before, we transform the records into physics signals: displacement, velocity, and acceleration. 
We filter data between 1~Hz and 15~Hz, which allows the inclusion of seismic human noise (4--14~Hz), microseismicity, and environmental noise (more than 1 Hz) \cite[e.g.,][]{Lecocq2020, Ojeda2021}. 
Then, we will have three files for each station containing the displacement, velocity, and acceleration. 
See Appendix \ref{app:noise} for more technical details. 

\section{Methods}
\label{sec:methods}

This section develops both the theoretical framework and the practical methodology for estimating the quality of a weak random source. 
We begin, in Section \ref{sec:sv}, by recalling the formal definition of an $epsilon$-Santha–Vazirani (SV) source and explaining why it cannot be applied directly to finite data. 
We then, in Section \ref{sec:epsilonh}, introduce a method for estimating $\epsilon_h$, the deviation from uniformity for a fixed history length, which serves as the basic building block of our analysis. 
Next, in Section \ref{sec:ecombine}, we address the problem of combining the sequence($\epsilon_h$) into a single effective parameter $\epsilon$. 
To this end, in Section \ref{sec:axioms}, we propose a set of axioms that any reasonable estimator should satisfy and analyze several functions fulfilling these criteria, identifying the most suitable choice for our purposes. 
We continue, in Section \ref{sec:averages}, by exploring several examples of such functions and identifying the most suitable one for our scenario. 
Finally, in Section \ref{sec:discretization}, we discuss discretization procedures and the final preprocessing step for mapping real-valued data to binary sequences.

\subsection{Formal definition}
\label{sec:sv}

Let us start by recalling the definition of $\epsilon$-SV-Source already introduced in Eq.\ (\ref{eq:sv_intro}).
\begin{definition}[$\epsilon$-Santha-Vazirani-Source \cite{SV}]
\label{def:sv}
    We say that the source $S$ (that produces some binary sequence $s_1, s_2, \ldots$) is $\epsilon$-Santha-Vazirani-Source if we have that
    \begin{equation}
        \begin{split}
            \mathop\forall_{n \in \mathbb{N}} &\quad \mathop\forall_{s_0, \dots, s_{n+1} \in  \{ 0, 1 \} }\\
            &\frac{1}{2} - \epsilon \leq P(S_{n+1} = s_{n+1} | S_{n} = s_{n}, \dots, S_0 = s_0, E)\\ 
            &\qquad\qquad\leq \frac{1}{2} + \epsilon
        \end{split}
    \end{equation}
    where $E$ represents all other random variables in the past light cone of $S_{n+1}$.
\end{definition} 
Note that, for $\epsilon = 0$ we obtain a fully random source, and for $\epsilon = 1/2$ the source can be even deterministic. 

In practice, this definition cannot be applied directly to empirical data. First, the side information $E$ is neither observable nor statistically accessible. Second, real-world estimation necessarily relies on a finite sequence of outputs. We therefore adopt a standard relaxation and remove the explicit conditioning on $E$, retaining the assumption that any correlations with the environment are bounded by $\epsilon$.

This step is crucial, but has to be carefully justified for any tested source. 
It is also very important to note that removing the random variable $E$ does not mean that the SV source is fully uncorrelated with the eavesdropper or the environment. 
This only means that the correlations are not too strong and are limited by the epsilon. 
It is indeed a much weaker assumption.

Now we will present separately, for each type of sources, an argument justifying omission of the random variable $E$.

For the case of QRNG, we can assume that the source is trusted but not perfect. 
Therefore, we are not concerned by strong correlations from outside but still have to treat such QRNG as an SV source and not a fully random generator because of imperfection of quantum hardware.

In the events case (strong earthquakes), since it is unlikely that the adversary could influence the seismic signals, especially the one from the strong earthquake with an epicenter between 30~km and 1000~km below the Earth's surface. 

Seismic noise, by contrast, originates from a superposition of many natural and anthropogenic processes and cannot be assumed to be free from all adversarial influence. While large-scale coordinated human activity could, in principle, introduce structured components into the signal, any such influence is expected to be limited, indirect, and weakly correlated with the extracted bits. We therefore do not claim that seismic noise is a cryptographically secure randomness source, but rather model it as a weak random source whose correlations with the environment are bounded by $\epsilon$ under realistic threat assumptions.

\begin{remark}[Classical randomness amplification]
    Furthermore, when we talk about classical randomness amplification of two independent SV sources, we need them to be fully independent.
Therefore, our approach requires significantly less restrictive assumptions.
\end{remark}

Under the above assumptions, the SV condition reduces to:

\begin{equation}
    \begin{split}
        &\mathop\forall_{n \in \mathbb{N}} \quad \mathop\forall_{s_0, \dots, s_{n+1} \in  \{ 0, 1 \} }\\
        &\frac{1}{2} - \epsilon \leq P(S_{n+1} = s_{n+1} | S_{n} = s_{n}, \dots, S_0 = s_0) \leq \frac{1}{2} + \epsilon.
    \end{split}
\end{equation}
We can then further transform the inequality in the following way.
\begin{equation}
    \begin{split}
        &\mathop\forall_{n \in \mathbb{N}} \quad \max_{s_0, \dots, s_{n+1} \in 
        \{ 0, 1 \} } \\
        &\left| P(S_{n+1} = s_{n+1} | S_{n} = s_{n}, \dots, S_0 = s_0) - \frac{1}{2} \right|  \leq \epsilon.
    \end{split}
\end{equation}

Note that if a source satisfies the $\epsilon$--SV condition, then it also satisfies the $\epsilon'$--SV condition for any $\epsilon' > \epsilon$. Since our goal is to identify the smallest admissible value of $\epsilon$, we assume throughout that $\epsilon$ is chosen optimally. With this convention, we retain the same symbol and rewrite the definition as a supremum over all sequence lengths $n$.

\begin{equation}
    \begin{split}
        \epsilon = &\sup_{n \in \mathbb{N}} \quad \max_{s_0, \dots, s_{n+1} \in 
        \{ 0, 1 \} } \\
        &\left| P(S_{n+1} = s_{n+1} | S_{n} = s_{n}, \dots, S_0 = s_0) - \frac{1}{2} \right|.
    \end{split}
\end{equation}
To take into account finite-memory effects, we introduce a history length $h$ and define:

\begin{equation}
    \begin{split}
        \epsilon &\leq \sup_{n \in \mathbb{N}} \quad \max_{h \in \{ 0, \dots, n \}} \quad \max_{s_{n-h+1}, \ldots, s_{n+1} \in \{ 0,1 \}}\\
        &\left| P(S_{n+1} = s_{n+1} | S_{n} = s_{n}, \dots, S_{n-h+1} = s_{n-h+1}) - \frac{1}{2} \right|.
    \end{split}
\end{equation}
The above inequality is true, since we add maximization over a whole set of histories rather than a single one. 
We can even further enlarge the allowed history length by choosing the supremum instead of the previous maximum. 
\begin{equation}
    \begin{split}
        \epsilon &\leq \sup_{n \in \mathbb{N}} \quad \sup_{h \in \mathbb{N}} \quad \max_{s_{n-h+1}, \ldots, s_{n+1} \in \{ 0,1 \}}\\
        &\left| P(S_{n+1} = s_{n+1} | S_{n} = s_{n}, \dots, S_{n-h+1} = s_{n-h+1}) - \frac{1}{2} \right|.
    \end{split}
\end{equation}

Strictly speaking, the above definition implicitly assumes access to histories longer than those generated by a finite source. In practice, this poses no difficulty: for sufficiently large history lengths, additional past variables do not affect the conditional distribution of the next output. We therefore adopt this simplified notation without loss of generality. With this convention, we can exchange the order of the suprema and arrive at the following expression:

\begin{equation}
\label{eq:epsilontosplit}
    \begin{split}
        \epsilon &\leq \sup_{h \in \mathbb{N}} \quad \sup_{n \in \mathbb{N}} \quad \quad \max_{s_{n-h+1}, \ldots, s_{n+1} \in \{ 0,1 \}}\\
        &\left| P(S_{n+1} = s_{n+1} | S_{n} = s_{n}, \dots, S_{n-h+1} = s_{n-h+1}) - \frac{1}{2} \right|.
    \end{split}
\end{equation}

We now introduce a sequence of parameters $\epsilon_h$, where $h$ denotes the history length, i.e., the number of past variables appearing in the conditional probability.

\begin{equation}
\label{eq:ehfirst}
    \begin{split}
        &\epsilon_h \coloneqq \sup_{n \in \mathbb{N}} \quad \max_{s_{n-h+1}, \ldots, s_{n+1} \in \{ 0,1 \}}\\
    &\left| P(S_{n+1} = s_{n+1} | S_{n} = s_{n}, \dots, S_{n-h+1} = s_{n-h+1}) - \frac{1}{2} \right|.
    \end{split}
\end{equation}

It allows us to rewrite Eq.\ (\ref{eq:epsilontosplit}) in the following form.

\begin{equation}
    \label{eq:sup}
    \epsilon \leq \sup_{h \in \mathbb{N}_0}\epsilon_h.
\end{equation}

We assume that the values of $\epsilon_h$ can be estimated for a suitable range of history lengths $h$. The estimation procedure is described in detail in the next section. Given these estimates, we can then define an effective parameter $\epsilon$ by appropriately combining the sequence $(\epsilon_h)$, as discussed below.

Although the definition of $\epsilon_h$ stated in Eq.\ (\ref{eq:ehfirst}) is formally correct, it presents two fundamental challenges in the finite-data regime. First, empirical estimation provides access only to a finite subset of the parameters $\epsilon_h$, namely for $h \in \{0,\dots,h_{\max}\}$, rather than the infinite sequence required by the definition. At first sight, one might attempt to address this limitation by approximating $\epsilon$ with a truncated estimator $\tilde{\epsilon}$ defined in an analogous manner:

\begin{equation}
    \tilde{\epsilon} = \max_{h \in \{ 0, \dots, h_{\mathrm{max}} \}} \epsilon_h.
    \label{eq:final_epsilon}
\end{equation}

Unfortunately, a second limitation arises from the fact that taking a maximum over history lengths is not suitable in practice. The statistical reliability of the estimates $\epsilon_h$ decreases rapidly with increasing $h$, making the maximum overly sensitive to poorly estimated terms. This motivates the need for an alternative aggregation rule to replace the truncated supremum. We address this problem in in the further sections. 
Before doing so, we first detail in the next section the procedure used to estimate the individual parameters $\epsilon_h$.

\subsection{Calculating epsilons for given history length}
\label{sec:epsilonh}

In this section, we explain how the history-dependent quantities $\epsilon_h$ introduced in the previous section can be estimated from a finite binary sequence. While the definition of $\epsilon_h$ is given in terms of conditional probabilities of an idealized source, experimental access is limited to a single realization of the process. We therefore replace these probabilities with empirical frequencies of short substrings, yielding practical estimators that can be computed directly from data. This construction forms the core of our randomness-testing procedure and underlies the implementation of the SVTest software. We also discuss intrinsic limitations of this estimation, in particular, the loss of statistical reliability as the history length increases.

Let us begin by recalling the definition of the $\epsilon_h$ already stated in the previous section.
\begin{equation}
    \begin{split}
        &\epsilon_h \coloneqq \sup_{n \in \mathbb{N}} \quad \max_{s_{n-h+1}, \ldots, s_{n+1} \in \{ 0,1 \}}\\
    &\left| P(S_{n+1} = s_{n+1} | S_{n} = s_{n}, \dots, S_{n-h+1} = s_{n-h+1}) - \frac{1}{2} \right|.
    \end{split}
\end{equation}
We will start by rewriting the formula using the definition of conditional probability, obtaining that
\begin{equation}
    \begin{split}
        &\epsilon_h \coloneqq \sup_{n \in \mathbb{N}} \quad \max_{s_{n-h+1}, \ldots, s_{n+1} \in \{ 0,1 \}}\\
    &\left\lvert \frac{P(S_{n+1} = s_{n+1}, S_{n} = s_{n}, \dots, S_{n-h+1} = s_{n-h+1})}{P(S_{n} = s_{n}, \dots, S_{n-h+1} = s_{n-h+1})} - \frac{1}{2} \right \rvert.
    \end{split}
\end{equation}

The numerator and denominator in the above expression involve probabilities of random variables that are not directly accessible experimentally. In practice, only a single finite realization of the source is available. We therefore estimate these probabilities using empirical frequencies of sub-strings occurring in the observed sequence, as defined below.

\begin{equation}
    \label{eq:epsilon_h_aprox}
	   \tilde{\epsilon}_h \approx
	   \max_{v_{h+1}} \left \lvert \frac{\frac{|s|_{v_{h+1}}}{n-h}}{\frac{|s|_{v_{h+1}'}}{n-h+1}} - \frac{1}{2} \right \rvert\ ,
\end{equation}
where $s$ is a finite binary string produced by the device that we are testing, the maximum over $v_{h+1}$ is taken over all possible binary strings of length $h+1$, and $|s|_{v_{h+1}}$ counts the number of occurrences of the string $v_{h+1}$ in the sequence $s$.
Furthermore, $v'_{h+1}$ is the sub-string of the string $v_{h+1}$ obtained by removing the first bit.

The key idea is to replace inaccessible probabilities with empirical frequencies computed from the observed sequence. Since only a single finite realization is available, we treat each bit of the tested string $s$ as a potential ``current'' output and examine its recent history of length $h$. In this construction, $|s|_{v_{h+1}}$ counts the number of occurrences of a specific length-$(h+1)$ substring $v_{h+1}$, corresponding to a given history followed by a particular output bit, while $|s|_{v'_{h+1}}$ counts the number of occurrences of the associated history substring. By construction, $|s|_{v'_{h+1}} = 0$ implies $|s|_{v_{h+1}} = 0$; in this case, the corresponding contribution is defined to vanish, ensuring that the estimator is well defined for all sequences.

Additionally, for a sufficiently large sequence length $n$ and fixed, small history length $h$, this expression simplifies, yielding the final estimator for $\epsilon_h$:
\begin{equation}
\label{eq:eh}
	   \tilde{\epsilon}_h \mathrel{\mathop{\approx}\limits_{n\to \infty}}
	   \max_{v_{h+1}} \left \lvert \frac{|s|_{v_{h+1}}}{|s|_{v_{h+1}'}} - \frac{1}{2} \right \rvert.
\end{equation}

The estimation, defined by the above formula, is realized by the central part of our software. 

In the next section, we will address the problem of combining a just-defined sequence of $\tilde{\epsilon}_h$ into a single final $\epsilon$.

\begin{remark}[Quality of $\tilde{\epsilon}_h$ estimation]
\label{rem:eh}
    It is important to point out here two important issues:
    \begin{itemize}
        \item For tested string $s$ with fixed length $n$, the quality of estimation drops down as the history length $h$ is increasing. 
    The above follows because the number of occurrences of each binary string will decrease with $h$ since the maximization in Eq.\ (\ref{eq:eh}) is taken over all binary strings of length $h$.
    \item If the string $s$ will be too short to contain all binary substrings of length $h$, the estimation will become trivial.
    \end{itemize}
    Because of that, when estimating the final $\epsilon$, we can take into account only $\epsilon_h$ for some reasonably small values of $h$, and additionally, we should incorporate them with decreasing weights.
\end{remark}

\subsection{Epsilons combining functions}
\label{sec:ecombine}

The estimation procedure described in the previous section yields a family of parameters $\epsilon_h$, each quantifying deviations from uniformity at a fixed history length. However, practical applications of randomness amplification and security analysis require a single figure of merit characterizing the source as a whole. Combining the sequence ($\epsilon_h$) into a single effective parameter $\epsilon$ is therefore a central challenge, particularly in the finite-data regime where large history lengths are poorly estimated. In this section, we formalize the problem of constructing such an aggregated estimator and explain why the direct formulation of the SV definition cannot be used directly in practice, motivating the need for alternative combination rules.

In \cite{HeartSV}, we proposed to use the weighted average in the following form
\begin{equation}
    \label{eq:HeartSVepsilon}
    \tilde\epsilon(s_n) \coloneqq \frac{1}{w(\lfloor\log_2(n)\rfloor-1)}\sum_{i=0}^{\lfloor\log_2(n)\rfloor-1} \frac{\tilde\epsilon_i(s_n)}{(i+1)}
\end{equation}
with $w(h)=\sum_{i=0}^{h} \frac{1}{i+1}$.
Although using the above formula is a reasonable choice, we did not justify it in the previous publication. 

We therefore analyze in detail the problem of constructing a single effective parameter $\epsilon$, which is central to this work. While the estimation of $\epsilon_h$ for fixed history length is straightforward, combining these quantities into a unique final value is inherently nontrivial. As discussed in previous sections
, the supremum over history lengths---although appropriate in the infinite-data limit---is unsuitable for finite sequences. This necessitates the development of alternative aggregation methods tailored to the finite-data regime.

We will first formalize the notation of the epsilon-generating function $\Psi$ and show that the direct analog of the $\epsilon$-SV source condition is not applicable in a finite testing context. 
Next, in the following section, we will provide a list of axioms (properties) that the reasonable epsilon-generating function should fulfill. 
Next, we will give a few examples of such functions and justify our choice of one of them.

We will start with a formal definition of the function that transforms the sequence of epsilons into a single final epsilon. 
Let $s=(s_i)_{i=1}^n$ be the sequence of $n$ bits obtained from a source using the chosen discretization method and optionally some additional post-processing (see Section \ref{sec:discretization}).
Then, as we described in the previous section, our software is capable of estimating, from the source $s$, the sequence of epsilons $(\epsilon_h)_{h=0}^{h_{\mathrm{max}}}$, each for given history length $h$, where $h_{\mathrm{max}} \leq n - 1$.

In the context of our estimation procedure, the admissible history length can be further restricted. For a parameter $\epsilon_h$ to take a value strictly smaller than $1/2$, it is not sufficient to require $h < n$. If a particular history never occurs in the sequence, the corresponding estimate trivially equals $1/2$. Meaningful estimation therefore requires that all histories of length $h$ can occur at least once. This condition yields the tighter bound that was previously derived in \cite{HeartSV}:

\begin{lemma}[Lemma 1 from \cite{HeartSV}]
For the sequence of data $s_n$ of length $n$ , the maximal length of history $h$ that satisfies
$\epsilon$-SV condition with $\epsilon < 1/2$ satisfies
\begin{equation}
    h \leq \log(n) - 1.
\end{equation}
\end{lemma}
The proof of this lemma and a further justification can be found in \cite{HeartSV}.
We will use that bound in this work, adjusting it to the manuscript notation, obtaining:
\begin{equation}
h \leq h' := \lfloor \log_2 n \rfloor - 1.
\end{equation}

Consequently, it suffices to consider aggregation functions $\psi_h$ with $h \leq h'$.
Moreover, since estimates associated with larger history lengths suffer from reduced statistical reliability, aggregation rules should assign greater weight to parameters corresponding to smaller values of $h$.

Let us define the function
\begin{equation}
    \Psi : l^\infty \to \left[0, \frac{1}{2}\right].
\end{equation}
to be the most general form of epsilon combining function that takes sequence  $(\epsilon_i)_{i=0}^{\infty}$ and output single final epsilon value.
We will also define the sequence of functions 
\begin{equation}
    \Psi_h : l^\infty \to \left[0, \frac{1}{2}\right].
\end{equation}
in such a way that each $\Psi_h$ depends only on the first $h$ arguments, namely
\begin{equation}\label{eq:Psipsi}
    \Psi_h((\epsilon_i)_{i=0}^{\infty}) \coloneqq \psi_h((\epsilon_i)_{i=0}^{h})
\end{equation}
for some function $\psi_h : [0, 1/2]^{h+1} \to [0, 1/2]$.

In the ideal case, we would like to $\lim_{h \to \infty} \Psi_h = \Psi$ for some mode of convergence \cite{rudin1976principles,knopp1990theory}, such as pointwise or uniform, that we will not specify here and also that $\Psi$ would be one defined as supremum in Eq.\ (\ref{eq:sup}). 

Instead of imposing that convergence, we will construct a set of axioms for the function to fulfill. 

\subsection{Axiomatic approach}
\label{sec:axioms}

To resolve the ambiguity in aggregating the sequence of history-dependent parameters $\epsilon_h$ into a single effective measure of randomness, we adopt an axiomatic approach. Rather than selecting an ad hoc specific combination rule, we identify a set of minimal, physically motivated requirements that any reasonable estimator should satisfy. These axioms capture basic consistency, monotonicity, and robustness properties expected of a parameter intended to quantify weak random source from finite data. They provide a principled framework for comparing different aggregation schemes and serve as criteria for constructing estimators suitable for practical randomness testing.

Recall that a given data set of length $n$, epsilons with history length bigger than $n$ do not contain any information (they are always equal to 0.5).
Therefore, without loss of generality, we can restrict our analysis to the sequence of functions $\Psi_h$ for $h \leq n$.
Furthermore, based on Eq.\ (\ref{eq:Psipsi}), it is enough to investigate only functions $\psi_h$ that have a finite number of arguments. 

Having established mathematical language to work with, we will enumerate some necessary and desirable properties that the functions $\psi_h$ should satisfy.

We will state them as a set of axioms.
We use name axioms for this set of definitions describing desired properties to highlight their importance and state that they can be used to set the theoretical background for the estimation of the single epsilon.
Furthermore, we give four separate axioms even though they can be combined into a smaller set.
We see this distinction as useful since one can try to drop one or more of the axioms in the future to investigate a broader set of epsilon-combining functions. 

To shortly sum up the idea of axioms. There are desired properties that any reasonable epsilon-generating function should fulfill. They are based on both the general theoretical definition of $\epsilon$-SV Source and finite effects that come from our method of calculating a series of $\epsilon_h$ for finite tested data. The axioms will allow us to restrict the set of possible epsilon-generating functions. Furthermore, at a later stage, we will select specific types of epsilon-generating functions for our software that fulfill all the axioms.

\begin{axiom}[Zero condition]
\label{axiom:zero}
    Each function $\psi_h$ has to be equal to zero if the input is a sequence of only zeros
    \begin{equation}
        \psi_h(0, 0, \dots, 0) = 0.
    \end{equation}  
\end{axiom}

The first axiom makes sure that if all estimated $\epsilon_h$ are zeros, then the value of the final epsilon is also zero.
That condition is self-explanatory, although it should be considered together with the previous comments that we take into account only a finite number of $\epsilon_h$ up to some $h_\mathrm{max}$.
Additionally, one could strengthen this axiom by also imposing its reverse.

\begin{axiom}[Monotonicity]
\label{axiom:monotonicity}
    Each function $\psi_h$ has to be monotone for all of the variables, namely
    \begin{align}
        &\mathop\forall_{\epsilon_0, \dots, \epsilon_h}    \mathop\forall_{i \in \{0, \dots, h\}}  \mathop\forall_{\epsilon'_i > \epsilon_i}  \psi_h(\epsilon_0, \dots, \epsilon_{i-1}, \epsilon'_i, \epsilon_{i+1}, \dots, \epsilon_h) \nonumber\\
        &\geq \psi_h(\epsilon_0, \dots, \epsilon_{i-1}, \epsilon_i, \epsilon_{i+1}, \dots, \epsilon_h).
    \end{align}
\end{axiom}

The second axiom ensures that an increase in any $\epsilon_h$ cannot lead to a decrease in the final epsilon. 
The functions that do not fulfill this axiom are clearly against the spirit of the single-epsilon estimation.
Also, here, it is possible to strengthen this axiom by imposing strong inequality.

\begin{axiom}[Position influentiality]
\label{axiom:influence}
    \begin{equation}
    \begin{split}
        &\mathop\forall_{\epsilon_0, \dots, \epsilon_h}    \mathop\forall_{i,j \in \{0, \dots, h\}: i < j}  \mathop\forall_{\delta > 0} \\
        &\psi_h(\epsilon_0, \dots, \epsilon_{i-1}, \epsilon_i - \delta, \epsilon_{i+1}, \dots, \epsilon_{j-1}, \epsilon_j + \delta, \epsilon_{j+1}, \dots, \epsilon_h)
        \nonumber\\ 
        &\leq \psi_h(\epsilon_0, \dots, \epsilon_{i-1}, \epsilon_i, \epsilon_{i+1}, \dots, \epsilon_{j-1}, \epsilon_j, \epsilon_{j+1}, \dots, \epsilon_h)
    \end{split}
    \end{equation}
    where the quantifier for all $\delta$ means here for all $\delta$ that makes sense, i.e., such that added to or subtracted from specific epsilon do not extend $[0, 1/2]$ interval: $\epsilon_i -\delta \geq 0$ and $\epsilon_j +\delta \leq \frac12$ 
    The idea behind this axiom is that epsilon with a longer history should have less influence than previous ones (in the spirit of Remark \ref{rem:eh}).   
\end{axiom}

The third axiom, despite its complicated formulation, has a simple meaning. 
We demand here that $\epsilon_h$ with a smaller $h$ index should not have a smaller impact on the final epsilon than the one with a bigger $h$ index.
It is in the spirit of the already mentioned decrease in estimation accuracy when the history length increases.
Once again, one could make this axiom stronger by imposing strong inequality.

\begin{axiom}[Normalization]
\label{axiom:normalization}
    Let $a \in (0, 1/2]$ then 
    \begin{equation}
        \psi_h(a, a, \dots, a) = a.
    \end{equation}
\end{axiom}

Finally, the fourth axiom reflects the fact that if all of the appropriate $\epsilon_h$ are equal, then the final epsilon should have the same value.
Although, as we stated in the introduction, this axiom could be easily combined with the first one, we want to separate them for a few reasons.
Firstly, contrary to the first axiom, an attempt to strengthen the fourth axiom by imposing its inverse is not reasonable. 
Secondly, the fourth axiom is the least important one and is the first candidate to omit if one would like to allow a broader set of epsilon-combining functions.
Nevertheless, we strongly recommend not omitting it at all but rather replacing it with some weaker condition so that, for example, the range of the final epsilon is still correct.
Third, in the case of weighted averages, which we will focus on in the next section, we will see that each of the four axioms imposes different conditions on these averages.

\subsection{Weighted averages}
\label{sec:averages}

Having established a set of axioms constraining admissible aggregation rules, we now consider concrete estimators that satisfy these requirements. We focus on weighted averages of the history-dependent parameters $\epsilon_h$, which provide a simple, transparent, and computationally efficient means of combining information from different history lengths. By choosing weights that decrease with the history length, these estimators naturally reflect the diminishing statistical reliability of $\epsilon_h$ for large $h$. We analyze several such weighting schemes and identify the best-suited ones for practical randomness estimation and implementation in the SVTest software. 
Namely, we will consider weighted averages of the form
\begin{equation}
\label{eq:weighted}
    \epsilon = \sum_{h = 0}^{h_{\mathrm{max}}} w_h \epsilon_h
\end{equation}
where $w_h$ are some weights. 

In the Observation below, we provide sufficient conditions for the weights to fulfill Axioms  \ref{axiom:zero}--\ref{axiom:normalization} given in the previous section.

\begin{observation}
\label{obs:axioms}
    The weighted average from Eq.\ (\ref{eq:weighted}) with positive, non-increasing, and normalized weights fulfills Axioms \ref{axiom:zero}--\ref{axiom:normalization}. 
\end{observation}
\begin{proof}

    Axiom \ref{axiom:zero} is true for all weighted averages based on the definition given in Eq.\ (\ref{eq:weighted}).
    Additionally, if all of the weights are nonzero, the conversion of Axiom \ref{axiom:zero} is also true, nevertheless, we do not necessarily impose that. 
    Furthermore, to fulfill the other axioms, we need to impose some additional conditions on the weights.
    If all weights are positive, then the Axiom \ref{axiom:monotonicity} is fulfilled.
    If weights are a non-increasing sequence, then Axiom \ref{axiom:influence} is also fulfilled.
    Finally, if we only allow normalized averages, we will satisfy Axiom \ref{axiom:normalization}.
\end{proof}

Now, we will consider two types of such weighted averages, starting from the one presented in our previous work \cite{HeartSV}.
A natural generalization of the weights $1/(i+1)$ are their powers, i.e., weights $1/(i+1)^{k}$ for some fixed natural number $k$ (for details of implementation, see Section \ref{sec:svtest}). 
Another possible choice that is well justified is to take weights $\frac{1}{2^i}$. 
The latter approach is very reasonable, since the exponential weights correspond to the probability of occurring in history a string of length $i$ with the assumption of its uniform distribution.

Additionally, we should note that for a low number of epsilons that form the final epsilon, the powers of $(i+1)$ can yield a lower value than $2^i$ (indeed, the exponential function is larger asymptotically than the polynomial one, while for low values, the polynomial can be larger). 
However, we prefer the exponential weights because of the theoretical justification mentioned above.

It is also important to note that both types of weighted averages mentioned above (with appropriate normalization) fulfill the assumptions of Observation \ref{obs:axioms}. 

Now we are ready to formally define two types of weighted averages $\epsilon_{\mathrm{poly},k}$ and $\epsilon_{\mathrm{exp}}$ that we started introducing above.

\begin{equation}
    \label{eq:epsilon_poly}
 \epsilon_{\mathrm{poly},k} \coloneqq \frac{1}{\sum\limits_{i = 0}^{h_{\mathrm{max}}} (i +1)^{-k}} \sum_{h = 0}^{h_{\mathrm{max}}} \frac{\tilde{\epsilon}_h}{(h+1)^k}.
\end{equation}
The above definition generalizes the one given in \cite{HeartSV} to higher powers $k\geq 1$.

In the case of exponential weights $\frac{1}{2^i}$, we obtain
\begin{equation}
 \epsilon_{\mathrm{exp}} \coloneqq \frac{1}{\sum\limits_{i = 0}^{h_{\mathrm{max}}} 2^{-i}} \sum_{h = 0}^{h_{\mathrm{max}}} \frac{\tilde{\epsilon}_h}{2^h}.
\end{equation}
Here, the normalization can be expressed as the sum of a geometric series; hence, we end up with the following form: 
\begin{equation}
    \epsilon_{\mathrm{exp}} \coloneqq \frac{1}{2-2^{-h_\mathrm{max}+1}} \sum_{h = 0}^{h_{\mathrm{max}}} \frac{\tilde{\epsilon}_h}{2^h}.
    \label{eq:svtest_epsilon}
\end{equation}

It is not hard to notice that the mean number of occurrences of strings of length $i$ (the current bit and its history of length $h = i-1$) is equal to $k=\frac{n}{2^{i}}$. 
It decreases with increasing $i$ (for the maximal $i=\lfloor \log n \rfloor\,$ $k=1$ ). 
Hence, for large enough $i>i_0$, with high probability, some string does not appear and $\epsilon_{i}=\frac{1}{2}$. 
However, in such a case, the exponential weights of $\epsilon_{i}$ for $i > i_0$ are relatively small, which makes them irrelevant for the final value of the $\epsilon$.

The SVTest software implements both $\epsilon_{\mathrm{poly},k}$ and $\epsilon_{\mathrm{exp}}$; however, in experiments, due to the above reasons, we use the $\epsilon_{\mathrm{exp}}$ as in Eq.\ (\ref{eq:svtest_epsilon}).

\subsection{Discretization}
\label{sec:discretization}

As described in the previous Section \ref{sec:data}, obtained and preprocessed earth data are in the form of sequences of real numbers \(d = (d_i)_{i=0}^{l} \) where \( \forall_i d_i \in \mathbb{R} \). 
On the other hand, we wish to have a sequence of bits \( s = (s_i)_{i=0}^{n} \) where \( \forall_i s_i \in \{ 0, 1 \} \) that we can further test, use, or amplify its randomness. 
Generally, we could arbitrarily choose the length parameter \( n \) and use any deterministic function \( \delta : d \to s \).
However, for practicality and simplicity reasons, we will limit ourselves to the case where \( n = l \) (the length of row data is equal to the length of the preprocessed data) and the discretization function is ``local''.  
By local, we mean that each bit \( s_i \) is computed in the same way and depends only on some small neighborhood of input numbers \( \{ d_{i-r}, \dots, d_{i+r} \} \). 
The main idea of this restriction is that the discretization should create \(i\)-th bit in the sequence directly from \(i\)-th real number or from the relation between \(i\)-th real number and its up to $r$ predecessors and successors.

We should mention here that a careful reader could notice that, in fact, the discretizations \ref{disc:2} and \ref{disc:3} are not ``local'' in the strict sense. 
Nevertheless, averages used in these discretizations can be seen as metaproperties of the source, not as direct dependence on all bits. 
Furthermore, when using some type of source regularly, we could try to estimate these averages in advance and treat them as a constant value for future runs. 
For example, our results for discretization \ref{disc:2} indicate that the average bit value is extremely close to zero. 
Therefore, we could assume that it is, in fact, zero in all future tests for the same source and the same kind of measurement apparatus. 
In the case of discretization \ref{disc:2}, it would make it equal to discretization \ref{disc:1}, which we observe in our results. 
In conclusion, in that broad sense, all methods of discretization presented in our manuscript can be seen as local.

In our experiments, we use several discretization methods described in the Appendix \ref{sec:discretization_details}, delineated as follows:

The first discretization attributes $0$ to a real number $d_i$ when $d_i>0$ and $1$ otherwise. Hence, it depends on the sign (see Definition \ref{disc:1}). 
The second discretization maps $d_i$ to $0$ when $d_i \geq \mathbb{E}[d]$, where
$\mathbb{E}[d] = \frac1n\sum_i d_i$, and maps to $1$ if it is not the case.  Hence, it depends on the average value over the whole sequence (see Definition \ref{disc:2}). 
The third discretization distinguishes from the second by replacing the condition defining the value of the output bit, namely $|d_i|\geq \mathbb{E}[|d|]$ where $\mathbb{E}[|d|]=\frac1n \sum_i |d_i|$ (see Definition \ref{disc:3}). As we will see this change significantly affects the value of final $\epsilon$.
The fourth discretization maps $d_i$ to $0$ iff $d_{i+1}\geq d_i$, that is
when the values increase from step $i$ to step $i+1$, and $1$ else (see Definition \ref{disc:4}). 
Finally, the fifth discretization maps $d_i$ to $0$ iff $|d_i|\geq |d_{i+1}|$, and $1$ else, i.e., like in the case of the fourth, but up to modulus. (see Definition \ref{disc:5}).

To finish this section, we will briefly discuss the case of using a min-entropy source instead of an SV source.
Min-entropy source is another commonly used definition of a weak random source. 
It has less structure than an SV source, and additionally, every SV source is also a min-entropy source, but not necessarily the other way around. 
Although in most applications, the approach to these two kinds of sources is fundamentally different, in our case of randomness estimation, we do not need any important modifications.
We will summarize this in the following remark, preceded by an observation of multiple runs of the device.

\begin{observation}[Indistingushability of multiple run of the device]
\label{obs:multiple_runs} There is no difference between $k$ subsequent runs of the device, each generating $n$ bits and one long run generating $kn$ bits.
\end{observation}

\begin{remark}[Estimating randomness of the min-entropy sources ($\mathrm{H}_{\min}$)]
    In the context of min-entropy sources, we can distinguish two types, the standard one-shot and the so-called block min-entropy source. 
    The block min-entropy source can be seen as a generalization of the SV source where we do not have separate single bits with history but a whole small sequence of bits with other sequences of bits as its history.
    In this case, our mathematical formulation and, through this, our software, can be modified to count frequencies of whole $k$ bit sequences where $k$ is the size of a block in the block min-entropy source.
    Although this modification is not currently implemented, it only requires changes in the frequency-counting part and does not heavily influence the rest of the software when various epsilons are counted.
    We should additionally mention here that with the increase in the block size, the required number of bits for a reasonable estimation of some history lengths increases drastically.
    Finally, when working with one-shot min-entropy sources, the estimation method cannot be different from the one described above for block min-entropy sources (see Observation \ref{obs:multiple_runs}). 
\end{remark}

\section{SVTest software}
\label{sec:svtest}

This section summarizes our SVTest Software’s architecture; the user can find further details in the ``README.md'' file of our SVTest Software \cite{SVTest}.

The program consists of three main stages: in the first stage, it uses two programs to download seismic data from accessible sources \cite{irismap} and outputs a .mseed file; in the second stage, a program transforms the data from ``.mseed'' format to ``.ascii'' format, and finally, in the third stage a program estimates the randomness parameters ($\epsilon_{h}$ and final $\epsilon$) from the ``.ascii'' input.

The first program of the first stage takes a list of seismic stations written in a ``.txt'' file, transforms it into the form required by the second program, and writes it in another ``.txt'' file. 
After this step, the software executes the second program, in which the user enters parameters determining the downloaded data and saves it in a ``.mseed'' file, the standard format for exchanging seismic
data. 

Later, in the second stage, we transform the downloaded data to a file more suitable for randomness source modeling; we proceed in two steps: We create separate ``.ascii'' files for every channel of every selected station and then aggregate the whole set of files into one ``.ascii'' file.

In the third stage, the program written in C language inputs data from the previous ``.ascii'' file and calculates the $\epsilon$ parameter of the potential SV source. 
This program provides the user with a few clear options to choose from:
\begin{enumerate}[label=(\alph*)]
    \item discretization method (see Section \ref{sec:discretization}),
    \item  the method of counting $\epsilon$ (see Section \ref{sec:averages}),
    \item history length (see Section \ref{sec:epsilonh}), and
    \item the number of lines taken from the final ``.ascii'' file, which is  equivalent to setting the number of initial seed bits.
\end{enumerate}

It is important to notice that the main part of our software, referred to above as the third stage, is universal and can be used with any weak random source, not only the one obtained from seismic events.
Furthermore, in the current version, it not only accepts data file in the text format but also binary ones.

In the next section we will describe in more detail the part of the program associated with the third stage since it is one of the main results of this work.

\subsection{Core of the SVTest program}

The goal of this program is to first estimate the sequence of values $\epsilon_h$ for the given history length $h$ and then estimate from them the final value of $\epsilon$.
The whole program is based on the mathematical previously discussed  background.

The first step is to load all input data into memory to allow faster computations.
Although this version does not support live streaming of data as an input, such a use case can be resolved by storing streamed data and dividing it into appropriate big parts to use in the software. 
If these parts are big, then the estimation error should be negligible.

The second part uses one of the discretization methods to obtain a bit sequence from the real number sequence used as the input data.
Our software implements a few different methods of discretization described in 
Appendix \ref{sec:discretization_details}.
Furthermore, each discretization is implemented as a separate function, so it is easy to modify it or create a new one without the need to change the other parts of the software. 
It could be useful if one would like to use the software to test some other source that requires some specific form of discretization.
Finally, if the data is already in the binary format, discretization could be omitted. 

The next part is the most crucial one:
We estimate the sequence of values $\epsilon_h$ for the given history length $h$ 
(according to the formula given in Eq.\ (\ref{eq:epsilon_h_aprox})).
The above is done by calculating the frequencies of appropriate substrings.
Since this part is the most computationally demanding, it is highly optimized by calculating each frequency only once (even if it is needed in more than one $\epsilon_h$).
Furthermore, we use low-level bit operations on substrings rather than calculating each substring frequency separately to improve efficiency even more. 

The last part is responsible for calculating the final value of $\epsilon$.
We implemented two families of weighted averages.
Namely, exponential average (see Eq.\ (\ref{eq:epsilon_poly})) and polynomial average (see Eq.\ (\ref{eq:svtest_epsilon})).
Here, the calculation is performed in a separate function, so it is easy to modify it or create a new one.
Therefore, any function that is coherent with the form described in Section \ref{sec:ecombine} can be used, although we recommend one that fulfills at least part of the axioms.

\section{Results}
\label{sec:results}

In this section, we will present numerical experiments results obtained from our software.
We will divide it into three parts depending on the kind of the source.
First, in Section \ref{sec:resultsQRNG}, we describe the results of the testing of a high-quality quantum random number generator.
Second, in Section \ref{sec:resultsNoise}, we present results for seismic noise.
Third, in Section \ref{sec:resultsEvents}, we provide results for earthquakes.
Finally, in Section \ref{sec:resultcomparision}, we discuss differences in the results obtained and their possible applications.

\subsection{Results for data from quantum random number generator}
\label{sec:resultsQRNG}

The first entropy source that has been tested was a self-certifying quantum random number generator SeQRNG provided by SeQure Quantum \cite{SeQureQuantum}. 
Since this device has inbuilt entropy estimation and randomness extraction procedures, its outcome has a guaranteed rate of 1 bit of min-entropy for every bit of outcome produced. 
Therefore, it can be considered a prefect randomness source and serves here as a benchmark.
\paragraph{Physical system.}
The SeQRNG is based on a quantum optical interferometer with active phase modulation. 
In each round, the outputs of two single-photon detectors
$D_0$ and $D_1$ are recorded. The random and unpredictable nature of
quantum measurement outcomes at the two detectors provides the raw
source of randomness.
\paragraph{Self-certification.}
A key feature of the device is that randomness generation and
entropy certification run concurrently. Two modes of operation are
randomly interleaved: \emph{randomness generation rounds}, in which
the detector outcomes are intrinsically unpredictable and form the
raw bit string, and \emph{self-certification rounds}, in which a
deterministic output is expected. Any deviation from the expected
outcome in the certification rounds signals a reduction in the entropy
of the randomness generation rounds, whether due to hardware
imperfections or adversarial interference. A  lower bound on
the min-entropy $H_{\min}$ of the raw output is continuously
computed from the observed statistics of both round types, under a
worst-case assumption that all deviations from ideal behavior are
attributed to an adversary.
\paragraph{Entropy extraction and security of source.}
By ``security of source'' we mean this continuously monitored,
certified lower bound on the min entropy of the raw QRNG data, established not by one-time
factory calibration but by the live self-certification protocol. The raw
bit string is post-processed using a NIST-vetted randomness
extraction protocol (HMAC with SHA-256 hashing)~\cite{NIST_FIPS_198,
NIST_FIPS_180}, which condenses the raw output into a shorter string
in which every output bit carries exactly 1 bit of
min-entropy. 

\paragraph{Use as a benchmark and the claim of ``perfect'' source.}
Because the raw min entropy estimation and subsequent extraction step guarantees 
$H_{\min} = 1$ bit per output
bit, the conditioned output saturates the maximum possible min-entropy
for a binary source. In the Santha--Vazirani framework used in this
paper, this corresponds to $\epsilon = 0$.

Table \ref{tab:qrng} presents the results of the SVTest software obtained from data from the quantum random number generator. 
The final epsilon values suggest that the generator is of very high quality and can be used directly in various protocols. 
Furthermore, epsilons for different history lengths $h$ clearly illustrate that the quality of estimation decreases with $h$ as stated in Remark \ref{rem:eh}.
Therefore, we stress here that with increasing history length our estimation of the security of source is getting less precise, not the security of source itself.

\subsection{Results for seismic noise}
\label{sec:resultsNoise}

Let us now discuss the efficiency of the seismic apparatus in generating partially random bits. 
This is an important problem since the previous approach via heartbeat suffered from low rates for natural reasons. 
The size of raw data downloaded from a particular apparatus may heavily depend on its location. 
We thus focus on the average size of filtered data from a package of raw data of fixed size.

Our detailed analysis first focuses on the noise, and the results are organized as follows. 
Tables \ref{tab:bits1Mb}--\ref{tab:bits1.5Gb} show values of  $\tilde{\epsilon}_h$ approximated as in Eq.\ (\ref{eq:eh}) for every discretization (see Appendix \ref{sec:discretization_details}). 
Each of the tables refers to different amounts of preprocessed data: $1$~Mb, $10$~Mb, $100$~Mb, $1$~Gb, and $1.5$~Gb, respectively. 
The symbol of three vertical dots means that for these values $\tilde{\epsilon}_h=\frac{1}{2}$, so in the tables $h$ ranges from $0$ to the minimal value of $h$ for which $\tilde{\epsilon}_h=\frac{1}{2}$. 
In the last row, there are values of $\epsilon$ defined in Eq.\ (\ref{eq:weighted}).
To make the analysis of data easier, the values from the last row are plotted both as a function of the number of preprocessed bits (Figures \ref{fig:bitsNo3} and \ref{fig:bitsOnly3}) and discretization types (Figures \ref{fig:discrNo3} and \ref{fig:discrOnly3}). 
Additionally, in Table \ref{tab:disc5} there are values of $\tilde{\epsilon}_h$  and $\epsilon$ for every number of preprocessed bits for discretization $2$ and Table \ref{tab:finalepsilon} gathers all values of  $\epsilon$ from Tables \ref{tab:bits1Mb}--\ref{tab:bits1.5Gb} for clear comparison. 

\begin{figure}[htbp]
    \centering
    \includegraphics[width=1\linewidth]{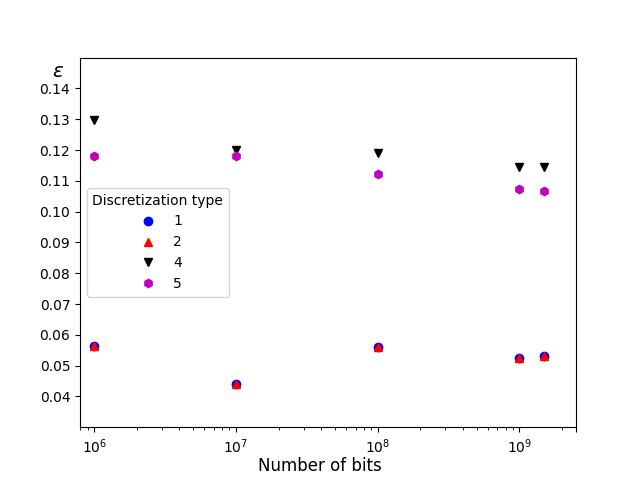}
    \caption{
        Various values of $\epsilon$ in terms of the initial number of seed bits for the given type of discretization. 
        The third discretization is beyond the scale and is presented in figure \ref{fig:bitsOnly3}.
    }
    \label{fig:bitsNo3}
\end{figure}

\begin{figure}[htbp]
    \centering
    \includegraphics[width=1\linewidth]{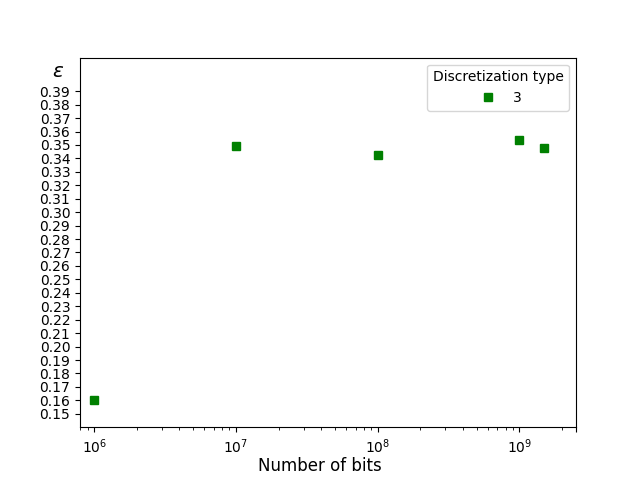}
    \caption{
        Various values of $\epsilon$ in terms of the initial number of seed bits for the third type of discretization.
    }
    \label{fig:bitsOnly3}
\end{figure}

\begin{figure}[htbp]
    \centering
    \includegraphics[width=1\linewidth]{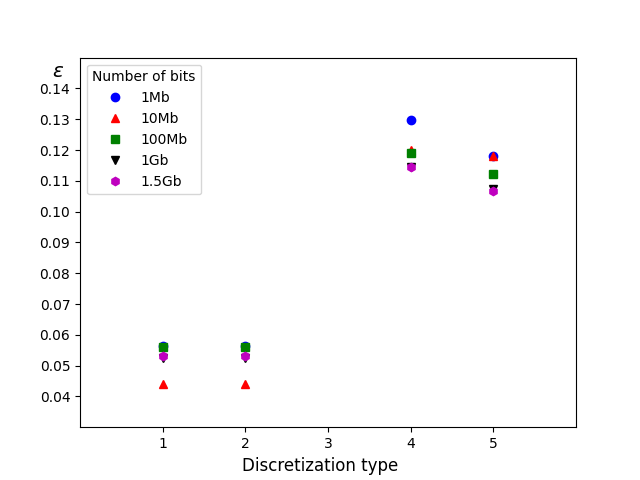}
    \caption{
        Various $\epsilon$ in terms of discretization for a given initial number of seed bits. 
        The third discretization is beyond the scale and has been presented in Figure \ref{fig:discrOnly3}.
    }
    \label{fig:discrNo3}
\end{figure}

\begin{figure}[htbp]
    \centering
    \includegraphics[width=1\linewidth]{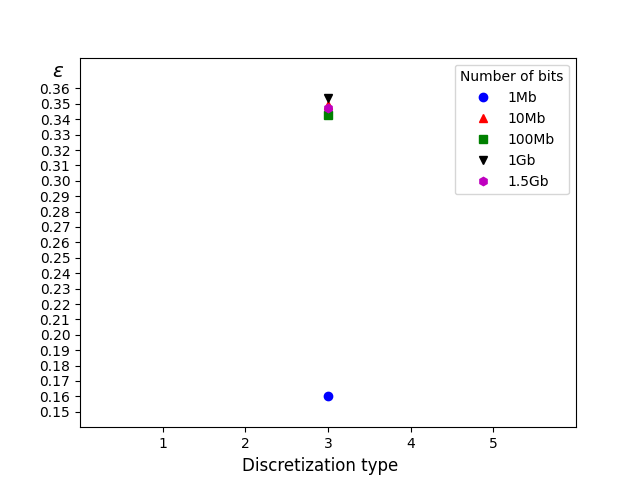}
    \caption{
        Various values of $\epsilon$ in terms of the third discretization for a given initial number of seed bits.
    }
    \label{fig:discrOnly3}
\end{figure}

The first study of values from Tables \ref{tab:bits1Mb}--\ref{tab:bits1.5Gb} reveals the conclusion that the choice of discretization has meaningful influence on the final value of $\tilde{\epsilon}_h$. It is easy to notice that the values for the first and the second discretization are optimal having one order of magnitude less than the others. Moreover, there is a significant difference in values of $\tilde{\epsilon}_h$ for the third discretization comparing them to the others. Their values are greater, although the value of $\tilde{\epsilon}_h$ gains $\frac{1}{2}$ for longer history length $h$. 
This concludes that applying the third discretization makes randomness weaker. Explanation of this fact is the following.
As one can see in Appendix \ref{sec:discretization_details}, bits are assigned according to the relation of the absolute value of the real numbers to the average value of the absolute values of the sequence. 
However, in such a way, some part of the information is lost, which causes the observed effect.

The next observation is that values of $\tilde{\epsilon}_h$  for the first and the second discretization are very close to each other. 
Notice that in the first discretization, the reference point of the values of the binary sequence is $0$, and in the second discretization, it is the mean value of the sequence of real numbers from the input file (see Appendix \ref{sec:discretization_details}). 
The method used to preprocess the seismic data (based on Fourier's transformation) causes this real numbers sequence to oscillate around $0$. 
Hence, the mean value is near $0$, which makes, in consequence, both discretizations' outputs practically the same.

\subsection{Results for seismic events}
\label{sec:resultsEvents}

The last Table \ref{tab:events} shows results for data obtained from deep and strong earthquakes that we refer to as events to distinguish them from the noise ones described above. 

As shown in the table, the estimated epsilon value ($\epsilon$) differed significantly across discretization type, signal type, and time-window type. The values range from reasonably small (around 0.04), to quite large (around 0.3). By analyzing these values, we can draw a few important conclusions. First, the signal response (velocity, acceleration, and displacement) and the time window type have some impact on epsilon, suggesting that different signal windows and preprocessing methods can contain different amounts of randomness. Second, the discretization type has a huge impact on the epsilon value, showing variation even to one order of magnitude at the same time window and signal response.
Furthermore, no single discretization is the best one.
Rather, the discretization method should be chosen for the signal response.
It also suggests that other discretization types should be developed and checked in the future.
Third, we can conclude that when the seismic event is carefully selected and appropriate discretization method is used, we can treat it as weak random source with quiet reasonable parameter epsilon.
The randomness amplification protocols, like \cite{Gallego2013} or \cite{Brando2016},  are suitable for weak random source with $\epsilon <\frac{1}{2}$. However, the closer to $0$ is this value the more efficient they are. Hence, the SV-Source with epsilon $0.04$, because of its optimal value among the others, can be used as the protocols input.

\subsection{Results comparison}
\label{sec:resultcomparision}

The value of $\epsilon$ is the reference point in the comparison of sources of randomness. Simply, one source of randomness is better than the other source, if $\epsilon$ value calculated for it is smaller comparing to the $\epsilon$ value calculated for the other source. The results comparison has been made for all three sources presented in the previous section. 
Due to the model, the quantum random generator is almost a perfect source of randomness, much better than seismic events. The values of $\epsilon$ for the former are a few orders of magnitude lower than for the latter (see Figure \ref{fig:QRNG_noise_comparison}). Hence, QRNG can be treated as a benchmark of the quality of our randomness estimation. Moreover, both types of sources could be useful in different scenarios: quantum number generator, because of its high quality, can be used as a direct source of randomness in contrast to seismic generators, which can be a potential input for quantum randomness amplification protocols. According to a given protocol's $\epsilon$ value requirement one can choose type of discretization and number of bits for which it is the most efficient.

The analysis and comparison of the results from earthquakes and noise do not reveal a difference in order of magnitude. Nevertheless, in terms of discretization types, the noise is better in the case of the first and second discretization for all types of signals and time windows used for earthquakes. The opposite situation concerns the third discretization: earthquakes are much better. For the remaining two discretizations (4 and 5) it depends on the type of signal time window (see Tables \ref{tab:events} and \ref{tab:bits100Mb}). Still, the best values of $\epsilon$ from earthquakes are better than the noise for these discretizations (see Figure \ref{fig:compNoiseNatural}). It is important to stress that the direct cause of influence on final $\epsilon$ value in relation to  chosen discretization  for seismic events was not in the area of the research. We leave it rather as an open question for possible future investigation.
Finally, when comparing seismic noise and seismic events, the estimated epsilon is only one factor. Another important aspect is the achievable rate of bit generation. It is hard to calculate, but intuitively, this rate is bigger for the noise. The argument for this claiming is the following: practically, the noise data is  gathered continuously, regardless of the order of magnitude of surface movement. On the other hand, data from natural events are more demanding; they concern earthquakes of specified magnitude and geological location, so are collected less often.

\begin{figure}[htbp]
    \centering
    \includegraphics[width=1\linewidth]{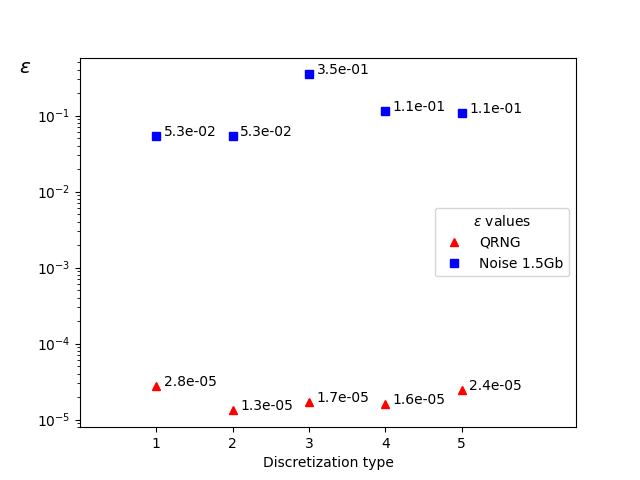}
    \caption{Comparison of values of $\epsilon$ in terms of discretization for QRNG and noise. The number of seed bits for noise seismic events is 1.5 Gb.
    }
    \label{fig:QRNG_noise_comparison}
\end{figure}

\begin{figure}[htbp]
    \centering
    \includegraphics[width=1\linewidth]{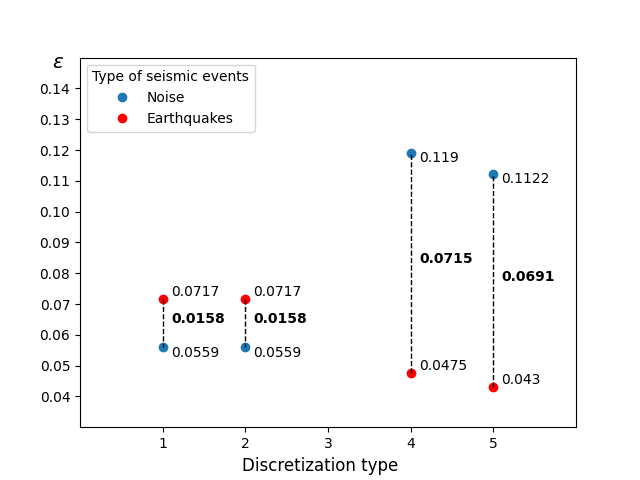}
    \caption{
        Comparison of values of $\epsilon$ in terms of discretization for earthquakes and noise. The number of seed bits for noise seismic events is 100Mb. For the earthquakes the smallest value has been chosen for a given discretization (see Table \ref{tab:events}). The length of line segment (written in bold) represents the difference in values of $\epsilon$ for both types of seismic events.  The third discretization is beyond the scale and has the following $\epsilon$ values  $0.3429$ for noise and $0.2354$ for earthquakes  giving $0.10754$ of difference.
    }
    \label{fig:compNoiseNatural}
\end{figure}


\section{Resource aware aspect}
\label{sec:avarenes}

We should also note that using the seismic apparatus is a very resource-aware and safe option for the environment in various ways. 

Firstly, it can use existing devices and technologies without developing and manufacturing completely new hardware. This approach reduces both the financial and environmental costs of weak randomness preparation. 

In addition, the devices used can serve dual purposes simultaneously, both as weak randomness sources and in standard seismic applications such as earthquake research or early warning systems. 
This synergy could benefit both fields. 
The places that have already implemented seismic devices could easily get access to the weak randomness. 
Furthermore, new seismic devices installed mainly for cryptography purposes could simultaneously be used as early warning systems that can benefit the general local population. 

Finally, contrary to some popularly used sources of randomness based on radioactive decay, this solution does not use potentially dangerous or restricted materials. 
This also reduces the environmental impact and inherent risk associated with the acquisition, handling, and disposal of radioactive or dangerous materials.

Last but not least, verifying whether a given weak random source can be modeled as an SV source has energy-saving aspects, in the spirit of the recently founded Quantum Energy Initiative \cite{QEI}.
That is, it is more common for providers of random number generators to assure the quality of the generated output based on a weaker model of the source, only assuming that it is the so-called $\mathrm{H}_{\min}$
source. 
The $\mathrm{H}_{\min}$ source is unstructured, as opposed to the SV source, which satisfies the so-called Generalized Chernoff bound \cite{Ramanathan2016}. 
Some of these providers, by using SVTest, might observe a higher generated randomness rate, or equivalently, less power use for the same amount of generated randomness than assuming the source has $\mathrm{H}_{\min}$ source structure. 
Investigating this expected phenomenon can be the topic of future research.

\section{Discussion}
\label{sec:discussion}

In this work, we address the problem of characterizing weak randomness in physical data by adopting the framework of Santha--Vazirani (SV) sources. In particular we have focused on verifying if a given source can be treated as the so-called SV source, parametrized by $\epsilon \in [0,1/2]$. 

In this article, we have contributed to the field of weak randomness analysis in three ways.
We first have developed the mathematical framework for estimating the quality of weak random source.
Namely we have introduced a principles-based method to estimate the SV parameter $\epsilon$, which quantifies the deviation of a source from ideal randomness. Indeed, since $\epsilon$ cannot be inferred uniquely from finite data, we first identify a set of natural axioms that any data-driven randomness estimation procedure should satisfy.

Second, building on this framework, we have developed \texttt{SVTest}, a software tool that estimates the SV parameter $\epsilon$ directly from experimental data. \texttt{SVTest} operates on numerical time series or binary data and can be readily adapted to other data formats. We release \texttt{SVTest} as open-source software~\cite{SVTest}, providing a practical and reproducible method for certifying weak randomness in physical sources.

Third, we apply \texttt{SVTest} to two classes of experimental data. As a benchmark, we first analyze a binary sequence generated by a commercial quantum random number generator provided by SeQure Quantum \cite{SeQureQuantum}. The extracted values of $\epsilon$ confirm that this device operates close to an ideal SV source. We then turn to seismic data, considering both earthquake signals and ambient seismic noise. Our analysis shows that seismic phenomena exhibit measurable, nontrivial SV randomness, motivating their study as physical weak random sources.

Our results indicate that seismic phenomena constitute public randomness sources that are plausibly not controllable by an adversary. Although the geographical regions where earthquakes are more likely to occur are known, the resulting seismic waveforms depend on a complex combination of source mechanisms, propagation effects, and local geological structures. Moreover, the earthquakes analyzed here originate several kilometers below the surface, placing them beyond any realistic possibility of direct manipulation with current technology. 

If satisfying the above conditions, deep seismic phenomena would provide the first concrete randomness sources to feed the most advanced techniques for \emph{amplifying and privatizing} randomness using quantum devices \cite{RotemKessler}.
We answer positively to the above by proposing seismic randomness sources of sufficient depth and certifying their suitability as SV sources. 
We achieve the above result by demonstrating that $\epsilon$ is distinctly smaller than 0.5 for the meaningful output bit sequences.

The case of seismic noise is, however, slightly different by its natures.
Since it is a combination of an enormous number of small signals from various origins, even if we can extract a little bit of information from the \cite[e.g][]{Ojeda2021, Lecocq2020}, as a whole, the sources of the noise is completely unknown \cite{Udias2014,stein2009,shearer2019}. We therefore do not claim seismic noise to be cryptographically secure, but instead treat it as a weak random source with environmental correlations bounded by~$\epsilon$ under realistic threat assumptions.

Finally, we will discuss possible future directions and open questions. 

We expect that an additional application of \texttt{SVTest} would be the benchmarking of sources that are declared to be random. 
Such sources are expected to operate close to the ideal case of an SV source with $\epsilon=0$. 
Since the set of $\epsilon$-SV sources forms a polytope generated by suitably permuted Bernoulli distributions \cite{Grudka2014}, one may test whether a given device can be certified as an $\epsilon_0$-SV source for some small $\epsilon_0$. 
Because of that, it is plausible that \texttt{SVTest} software could provide a practical tool for performing this certification directly from observed data, although this approach requires some further research.

A natural open question concerns the relation between the SV parameter $\epsilon$ and standard statistical indicators of randomness, such as $P$-values used in standard test suites (see \cite{NIST}). While $P$-values play a central role in empirical randomness testing, the parameter $\epsilon$ is the relevant quantity for applications such as classical and quantum randomness amplification. Establishing a quantitative connection between these two approaches remains an interesting direction for future work.

It would also be of interest to extend our approach beyond SV sources to other models of weak random sources, in particular min-entropy sources (for definitions, see, for example \cite{Bouda2014}). Min-entropy provides an alternative and widely used quantification of unpredictability, and understanding how SV certification methods relate to or can be adapted for min-entropy sources could further broaden the applicability of our framework.

From a practical perspective, future developments of \texttt{SVTest} software may include improved memory efficiency and the ability to analyze data streams in real time. Such extensions would further enhance the applicability of the software to experimental settings where continuous data acquisition and testing are required.

In summary, we have introduced a principled and operational framework for certifying weak randomness in physical data and demonstrated its effectiveness through the analysis of seismic signals. By combining an axiomatic approach to SV sources with a practical, open-source implementation, we show that deep seismic phenomena contain certifiable randomness suitable for quantum randomness amplification. Our results establish seismic data as a previously unexplored class of physical randomness sources and provide a concrete methodology for identifying and validating new resources for cryptographic and quantum-information applications.


\begin{acknowledgments}
    The authors would like to thank Paweł Horodecki	for very useful discussions and comments.
    The authors would like to thank Marcin Pawłowski and SeQure Quantum for providing data from quantum random number generator.
	RS acknowledges financial support by the Foundation for Polish Science through TEAM-NET project (contract no.\ POIR.04.04.00-00-17C1/18-00) and funding from the European Union’s Horizon Europe research and innovation programme under the project ``Quantum Secure Networks Partnership'' (QSNP, grant agreement No 101114043).
    KH acknowledges National Science Centre, Poland, grant Opus 25 No. UMO-2023/49/B/ST2/02468.
    CMY acknowledges the Fondecyt postdoctoral project 3220307.
	We acknowledge partial support by the Foundation for Polish Science (IRAP project, ICTQT, contract no.\ MAB/2018/5, co-financed by EU within Smart Growth Operational Programme). The 'International Centre for Theory of Quantum Technologies' project (contract no.\ MAB/2018/5) is carried out within the International Research Agendas Programme of the Foundation for Polish Science co-financed by the European Union from the funds of the Smart Growth Operational Programme, axis IV: Increasing the research potential (Measure 4.3). 
    All seismic data were downloaded through the EarthScope Consortium Web Services (https://service.iris.edu/). 
    The processing of this article benefited from various Python packages, including Obspy \cite{Obspy}. 
		
	


    \textbf{Software:} SVTest Software \cite{SVTest} and all other additional source codes are available on the GitHub repository: \url{https://github.com/DQI-UG/EarthSV}.

    \textbf{Competing interests:} The authors declare that there are no competing interests.

\end{acknowledgments}


\bibliography{EarthSV}


\appendix

\section{Data}
In this section, we will describe in more detail how we obtained and processed the data used in this work. Note that the seismic data are free and available in the databases we detail.

\subsection{Data from quantum random number generator}

We obtained around 20~GB of binary data from the quantum random number generator provided by Marcin Pawłowski and SeQure Quantum company \cite{SeQureQuantum}.
Due to memory constraints, we divided the data into five equal parts, each containing over 33 billions random bits.

\subsection{Seismic events}
\label{app:events}

For earthquake data, first, we obtain a list of seismic events using the Global CMT catalog \cite{CMT1, CMT2}. 
We searched the catalog for events with moment magnitudes between $\mathrm{Mw}=6.0$ and $\mathrm{Mw}=7.0$ and depths between 30~km and 1000~km between 1976 and 2021. 
We saved this information in a file for later use within the format of ``CMTSOLUTIONS'', saving 2218 events. 
We then extract the necessary information for each earthquake, such as location, initial time, and magnitude. 
Once we have the necessary information, we downloaded the seismic signal of each earthquake using \textit{MassDownloader} from the \textit{Obspy} module \cite{Obspy}. 
In particular, we downloaded all the available stations from the IRIS seismological service \cite{iris}. 
This procedure takes close to 170 hours and allows us to download data from 1729 events that contain information about 7114 stations and 426710 data files. 

Second, we process the data; this means taking the instrumental response of the data using an Obspy module, which allows us to transform the data from counts to physical magnitudes; then, we apply a detrend to eliminate the linear tendency, decimate, and interpolate the signal to have two samples per second, and cut the signal in the period designed.
Given the number of files, we divided the work into multiple jobs (eleven), obtaining displacement, acceleration, and velocity data for two different time windows mentioned in the main text. 
This process took around 177 hours. 
Then, we concatenated (united) the files from the previous division and obtained six files corresponding to two types of time windows for each signal response: displacement, acceleration, and velocity.

Finally, we removed broken lines (lines with non-numeric values) from the files. 
The first time window, including the body wave, contains 284282034 lines (numerical values) and has a file size of 7249134442~B. 
The second time windows corresponding to the surface waves contain 520407646 lines (numerical values) and have a file size of 13270355887~B.

\subsection{Noise}
\label{app:noise}

Apart from its natural origin, ground vibration can be caused by external events such as human or animal movement, traffic, etc. 
From seismic stations placed around the world, we have chosen a subset of apparatuses with detectors that are sensitive to such noise. 
The channels in these detectors should have two features: continuous recordings and a large number of samples per second (sps), such as $20$, $40$, $80$, or $100$. 
We have picked 362 points: Chile (7), Argentina (17), Iceland (7), Indonesia, Malaysia, and Australian External Territories close to them (36), Nepal and India (130), and the state of California in the USA (165). 
All of them are located near significant human clusters like metropolises or big cities (Santiago, Reykjavik, Singapore, Kathmandu, Los Angeles). 
The list of stations is taken from the IRIS website \cite{iris} and is available on the project GitHub repository \cite{SVTest} in the file ``gmap-stations.txt''. 
We have chosen two time ranges of the data gathered from the stations. 
The first was from 1 January to 31 January 2015, and the second was from 1 March to 22 April 2017. 
Both ranges have been divided into 24-hour periods.
Once we have the data, we eliminate the instrumental response by applying a 1--15~Hz filter to account for the environmental noise. 
Then, we eliminate the linear tendency and resample the data to four samples per second to have homogeneity in the data. 
All of this is possible thanks to the Python module Obspy \cite{Obspy}. 
After preprocessing 4.9~GB of raw data, we got 2050938429 bits of the seed.

\section{Discretization}
\label{sec:discretization_details}

As we have already introduced in Section \ref{sec:discretization}, the seismic data are in the form of a sequence of floating-point numbers.
To obtain the binary sequence, that is needed for our randomness test and also for any modern cryptographic application, we are using various discretization methods. 
In the following, we define the five discretization methods in detail.
The names of the discretizations refer to appropriate function names in the SVTest's source code for ease of reference.

\begin{definition}[discretizeEarthDataEvents1]
\label{disc:1}
	Let $d = (d_i)_{i=1}^n$ be a sequence of the input file values where $n$ is the number of these values. 
    Then the discretized binary sequence $s = (s_i)_{i=1}^n$ is defined as
	\begin{equation}
		s_i \coloneqq
		\begin{cases}
			0 : d_i \geq 0\\
			1 : d_i < 0
		\end{cases}.
	\end{equation}
\end{definition}

\begin{definition}[discretizeEarthDataEvents2]
\label{disc:2}
	Let $d = (d_i)_{i=1}^n$ be a sequence of the input file values where $n$ is the number of these values. 
    Then the discretized binary sequence $s = (s_i)_{i=1}^n$ is defined as
	\begin{equation}
		s_i \coloneqq
		\begin{cases}
			0 : d_i \geq \mathbb{E}d\\
			1 : d_i < \mathbb{E}d
		\end{cases}
	\end{equation}
	where 
	\begin{equation}
		\mathbb{E}d \coloneqq \frac{\sum\limits_{i=1}^{n}d_i}{n}
	\end{equation}
    is the average value of the sequence $d$.
\end{definition}

\begin{definition}[discretizeEarthDataEvents3]
\label{disc:3}
	Let $d = (d_i)_{i=1}^n$ be a sequence of the input file values where $n$ is the number of these values. 
    Then the discretized binary sequence $s = (s_i)_{i=1}^n$ is defined as
	\begin{equation}
		s_i \coloneqq
		\begin{cases}
			0 : \lvert d_i\rvert \geq \mathbb{E}\lvert d\rvert\\
			1 : \lvert d_i\rvert < \mathbb{E}\lvert d\rvert
		\end{cases}
	\end{equation}
	where 
	\begin{equation}
		\mathbb{E}\lvert d\rvert \coloneqq \frac{\sum\limits_{i=1}^{n}\lvert d_i\rvert}{n}
	\end{equation}
	is the average value of the absolute values of the sequence $d$.
\end{definition}

\begin{definition}[discretizeEarthDataEvents4]
\label{disc:4}
	Let $d = (d_i)_{i=1}^n$ be a sequence of the input file values where $n$ is the number of these values. 
    Then the discretized binary sequence $s = (s_i)_{i=1}^n$ is defined as
	\begin{equation}
		s_i \coloneqq
		\begin{cases}
			0 :  d_{i+1} \geq d_i\\
			1 : d_{i+1} < d_i
		\end{cases}.
	\end{equation}
\end{definition}

\begin{definition}[discretizeEarthDataEvents5]
\label{disc:5}
	Let $d = (d_i)_{i=1}^n$ be a sequence of the input file values where $n$ is the number of these values. 
    Then the discretized binary sequence $s = (s_i)_{i=1}^n$ is defined as
	\begin{equation}
		s_i \coloneqq
		\begin{cases}
			0 :  \lvert d_{i+1} \rvert \geq \lvert d_i \rvert\\
			1 : \lvert d_{i+1} \rvert < \lvert d_i \rvert
		\end{cases}.
	\end{equation}
\end{definition}

\section{Details of numerical results}

In this section, we gather the tables with numerical results obtained during our experiments.
Table \ref{tab:qrng} shows detailed results for the quantum random number generator described in Section \ref{sec:resultsQRNG}. 
In \cref{tab:bits1Mb,tab:bits10Mb,tab:bits10Mb,tab:bits100Mb,tab:bits1Gb,tab:bits1.5Gb,tab:disc5,tab:finalepsilon} we gathered results for seismic noise described in Section \ref{sec:resultsNoise}.
Finally, in Table \ref{tab:events} we present combined results for seismic events described in Section \ref{sec:resultsEvents}.

\begin{table}
    \centering
    \begin{tabular}{|c|c|c|c|c|c|}
        \hline
        &\multicolumn{5}{c|}{Discretization type}\\
        \hline
h&1&2&3&4&5\\\hline
0&0.0005690&0.0005660&0.1315510&0.0003220&0.0002130\\
1&0.0345043&0.0344991&0.1593988&0.1420655&0.1714941\\
2&0.1226864&0.1226757&0.1904835&0.3518887&0.2669886\\
3&0.2037178&0.2037109&0.2237833&0.3707880&0.3080776\\
4&0.2742169&0.2742294&0.2505766&0.4137014&0.3337402\\
5&0.3102384&0.3102700&0.2677215&0.4336691&0.3458401\\
6&0.3508916&0.3508916&0.2850156&0.4532433&0.3811881\\
7&0.3815717&0.3815717&0.3018971&0.4655244&0.4166667\\
8&0.4322034&0.4322034&0.3216523&0.5000000&0.5000000\\
9&0.4636364&0.4636364&0.3401715&$\vdots$&$\vdots$\\
10&0.5000000&0.5000000&0.3587355&&\\
11&$\vdots$&$\vdots$&0.3774831&&\\
12&&&0.3946779&&\\
13&&&0.4112256&&\\
14&&&0.4411765&&\\
15&&&0.5000000&&\\
$\vdots$&&&$\vdots$&&\\\hline
$\epsilon$&0.0564116&0.0564079&0.1601708&0.1298538&0.1179996\\\hline
    \end{tabular}
    \caption{Values of $\tilde{\epsilon}_h$ approximated as in Eq.\ (\ref{eq:eh}). Histories $h$ have length ranging from $0$ to the first one for which $\tilde{\epsilon}_h=\frac{1}{2}$. Results are obtained from $1$~Mb of preprocessed data, i.e., $1$~Mb bits that are the output of discretization. The types of discretization $1$--$5$ as defined in Section \ref{sec:discretization} are given in corresponding columns.}
    \label{tab:bits1Mb}
\end{table}

\begin{table}
    \centering
    \begin{tabular}{|c|c|c|c|c|c|}
        \hline
        &\multicolumn{5}{c|}{Discretization type}\\
        \hline
h&1&2&3&4&5\\\hline
0	&0.0000962	&0.0000962	&0.2521864	&0.0000309	&0.0000556	\\
1	&0.0240344	&0.0240344	&0.4121683	&0.1468744	&0.1742430	\\
2	&0.0885760	&0.0885760	&0.4695446	&0.2990353	&0.2541729	\\
3	&0.1388335	&0.1388335	&0.4886295	&0.3242693	&0.3278253	\\
4	&0.2579986	&0.2579986	&0.4956605	&0.3892212	&0.3301635	\\
5	&0.2787136	&0.2787136	&0.4974267	&0.4086667	&0.3528231	\\
6	&0.3415209	&0.3415209	&0.4983078	&0.4334802	&0.3631016	\\
7	&0.3691173	&0.3691173	&0.4987335	&0.4492071	&0.3945313	\\
8	&0.4076175	&0.4076175	&0.4989797	&0.4598069	&0.4629630	\\
9	&0.4184367	&0.4184367	&0.4991350	&0.4648480	&0.5000000	\\
10	&0.4740260	&0.4740260	&0.4992330	&0.5000000	&0.5000000	\\
11	&0.5000000	&0.5000000	&0.4992973	&$\vdots$	&$\vdots$	\\
12	&$\vdots$	&$\vdots$	&0.4993388	&	&	\\
13	&	&	&0.4993729	&	&	\\
14	&	&	&0.5000000	&	&	\\
$\vdots$	&	&	&$\vdots$	&	&	\\\hline
$\epsilon$	&0.0440134	&0.0440134	&0.3494201	&0.1199104	&0.1179383
\\\hline
    \end{tabular}
    \caption{Values of $\tilde{\epsilon}_h$ approximated as in Eq.\ (\ref{eq:eh}). Histories $h$ have length ranging from $0$ to the first one for which $\tilde{\epsilon}_h=\frac{1}{2}$. Results are obtained from $10$~Mb of preprocessed data, i.e., $10$~Mb bits that are the output of discretization. The types of discretization $1$--$5$ as defined in Section \ref{sec:discretization} are given in corresponding columns.}
    \label{tab:bits10Mb}
\end{table}

\begin{table}
    \centering
    \begin{tabular}{|c|c|c|c|c|c|}
        \hline
        &\multicolumn{5}{c|}{Discretization type}\\
        \hline
h&1&2&3&4&5\\\hline
0	&0.0000342	&0.0000342	&0.2646460	&0.0003795	&0.0000039	\\
1	&0.0349100	&0.0349101	&0.3905752	&0.1326140	&0.1632015	\\
2	&0.1215081	&0.1215081	&0.4394608	&0.3119693	&0.2498206	\\
3	&0.1856914	&0.1856914	&0.4569485	&0.3325120	&0.2974095	\\
4	&0.2890622	&0.2890622	&0.4667720	&0.3935813	&0.3279870	\\
5	&0.3146381	&0.3146381	&0.4721968	&0.4139880	&0.3513768	\\
6	&0.3782990	&0.3782990	&0.4761190	&0.4396989	&0.3578892	\\
7	&0.4178900	&0.4178900	&0.4788871	&0.4599814	&0.3627266	\\
8	&0.4587673	&0.4587673	&0.4809769	&0.4760734	&0.3690631	\\
9	&0.4664235	&0.4664235	&0.4826140	&0.4812527	&0.3933596	\\
10	&0.4761194	&0.4761194	&0.4839533	&0.4858114	&0.5000000	\\
11	&0.4837925	&0.4837925	&0.4850877	&0.4901020	&$\vdots$	\\
12	&0.5000000	&0.5000000	&0.4860544	&0.5000000	&	\\
13	&$\vdots$	&$\vdots$	&0.4869166	&$\vdots$	&	\\
14	&	&	&0.4876673	&	&	\\
15	&	&	&0.4883374	&	&	\\
16	&	&	&0.4889023	&	&	\\
17	&	&	&0.4894084	&	&	\\
18	&	&	&0.5000000	&	&	\\
$\vdots$	&	&	&$\vdots$	&	&	\\ \hline
$\epsilon$	&0.0559002	&0.0559002	&0.3428980	&0.1190001	&0.1121640
\\\hline
    \end{tabular}
    \caption{Values of $\tilde{\epsilon}_h$ approximated as in Eq.\ (\ref{eq:eh}). Histories $h$ have length ranging from $0$ to the first one for which $\tilde{\epsilon}_h=\frac{1}{2}$. Results are obtained from $100$~Mb of preprocessed data, i.e., $100$~Mb bits that are the output of discretization. The types of discretization $1$--$5$ as defined in Section \ref{sec:discretization} are given in corresponding columns.}
    \label{tab:bits100Mb}
\end{table}

\begin{table}
    \centering
    \begin{tabular}{|c|c|c|c|c|c|}
        \hline
        &\multicolumn{5}{c|}{Discretization type}\\
        \hline
h&1&2&3&4&5\\\hline
0	&0.0000379	&0.0000432	&0.2830588	&0.0002670	&0.0000039	\\
1	&0.0270989	&0.0271045	&0.3965429	&0.1237040	&0.1560276	\\
2	&0.1061174	&0.1061235	&0.4414523	&0.2979867	&0.2413325	\\
3	&0.1877859	&0.1877891	&0.4584365	&0.3253078	&0.2863797	\\
4	&0.2924943	&0.2924883	&0.4684143	&0.3900419	&0.3114722	\\
5	&0.3199377	&0.3199340	&0.4740149	&0.4115044	&0.3177254	\\
6	&0.3900256	&0.3900220	&0.4781338	&0.4414071	&0.3320376	\\
7	&0.4306055	&0.4306083	&0.4809785	&0.4625811	&0.3445827	\\
8	&0.4617990	&0.4617994	&0.4831211	&0.4765524	&0.3590555	\\
9	&0.4670683	&0.4670655	&0.4847453	&0.4791034	&0.3791076	\\
10	&0.4731727	&0.4731737	&0.4861209	&0.4820932	&0.4179695	\\
11	&0.4804570	&0.4804556	&0.4872526	&0.4870514	&0.4488332	\\
12	&0.4882257	&0.4882257	&0.4882274	&0.4907799	&0.4745401	\\
13	&0.4888889	&0.4888889	&0.4890809	&0.4965870	&0.4852129	\\
14	&0.5000000	&0.5000000	&0.4898071	&0.5000000	&0.4878935	\\
15	&$\vdots$	&$\vdots$	&0.4904634	&$\vdots$	&0.4938272	\\
16	&	&	&0.4909877	&	&0.5000000	\\
17	&	&	&0.4914608	&	&$\vdots$	\\
18	&	&	&0.4918838	&	&	\\
19	&	&	&0.4922764	&	&	\\
20	&	&	&0.4926278	&	&	\\ \hline
$\epsilon$	&0.0524899	&0.0524947	&0.3540503	&0.1143872	&0.1072149
\\\hline
    \end{tabular}
    \caption{Values of $\tilde{\epsilon}_h$ approximated as in Eq.\ (\ref{eq:eh}). Histories $h$ have length ranging from $0$ to the first one for which $\tilde{\epsilon}_h=\frac{1}{2}$. Results are obtained from $1$~Gb of preprocessed data, i.e., $1$~Gb bits that are the output of discretization. The types of discretization $1$--$5$ as defined in Section \ref{sec:discretization} are given in corresponding columns.}
    \label{tab:bits1Gb}
\end{table}

\begin{table}
    \centering
    \begin{tabular}{|c|c|c|c|c|c|}
        \hline
        &\multicolumn{5}{c|}{Discretization type}\\
        \hline
h&1&2&3&4&5\\\hline
0	&0.0000348	&0.0000384	&0.2732685	&0.0002804	&0.0000022	\\
1	&0.0289455	&0.0289494	&0.3923463	&0.1244670	&0.1536221	\\
2	&0.1059002	&0.1059046	&0.4384191	&0.2960425	&0.2395928	\\
3	&0.1875490	&0.1875526	&0.4562146	&0.3246583	&0.2864840	\\
4	&0.2961916	&0.2961876	&0.4667090	&0.3914730	&0.3168470	\\
5	&0.3242895	&0.3242866	&0.4725859	&0.4130765	&0.3195301	\\
6	&0.3947249	&0.3947223	&0.4768785	&0.4433095	&0.3369695	\\
7	&0.4349237	&0.4349255	&0.4798441	&0.4643555	&0.3499869	\\
8	&0.4648939	&0.4648947	&0.4820801	&0.4780088	&0.3625279	\\
9	&0.4697606	&0.4697591	&0.4837781	&0.4804853	&0.3821724	\\
10	&0.4748386	&0.4748390	&0.4852120	&0.4826861	&0.4112295	\\
11	&0.4822382	&0.4822382	&0.4863940	&0.4873306	&0.4437313	\\
12	&0.4874533	&0.4874533	&0.4874089	&0.4917997	&0.4695659	\\
13	&0.4926606	&0.4926606	&0.4883009	&0.4951574	&0.4834823	\\
14	&0.5000000	&0.5000000	&0.4890593	&0.5000000	&0.4853517	\\
15	&$\vdots$	&$\vdots$	&0.4897436	&$\vdots$	&0.4933659	\\
16	&	&	&0.4902970	&	&0.5000000	\\
17	&	&	&0.4907953	&	&$\vdots$	\\
18	&	&	&0.4912457	&	&	\\
19	&	&	&0.4916605	&	&	\\
20	&	&	&0.4920322	&	&	\\ \hline
$\epsilon$	&0.0531552	&0.0531586	&0.3474943	&0.1143968	&0.1066620\\\hline
    \end{tabular}
    \caption{Values of $\tilde{\epsilon}_h$ approximated as in Eq.\ (\ref{eq:eh}). Histories $h$ have length ranging from $0$ to the first one for which $\tilde{\epsilon}_h=\frac{1}{2}$. Results are obtained from $1.5$~Gb of preprocessed data, i.e., $1.5$~Gb bits that are the output of discretization. The types of discretization $1$--$5$ as defined in Section \ref{sec:discretization} are given in corresponding columns.}
    \label{tab:bits1.5Gb}
\end{table}

\begin{table}
    \centering
    \begin{tabular}{|c|c|c|c|c|c|}
        \hline
        &\multicolumn{5}{c|}{Preprocessed data}\\
        \hline
h&$1$ Mb&$10$ Mb&$100$ Mb&$1$ Gb&$1.5$ Gb\\\hline
0	&0.0005660	&0.0000962	&0.0000342	&0.0000432	&0.0000384	\\
1	&0.0344991	&0.0240344	&0.0349101	&0.0271045	&0.0289494	\\
2	&0.1226757	&0.0885760	&0.1215081	&0.1061235	&0.1059046	\\
3	&0.2037109	&0.2037109	&0.1856914	&0.1877891	&0.1875526	\\
4	&0.2742294	&0.2579986	&0.2890622	&0.2924883	&0.2961876	\\
5	&0.3102700	&0.2787136	&0.3146381	&0.3199340	&0.3242866	\\
6	&0.3508916	&0.3415209	&0.3782990	&0.3900220	&0.3947223	\\
7	&0.3815717	&0.3691173	&0.4178900	&0.4306083	&0.4349255	\\
8	&0.4322034	&0.4076175	&0.4587673	&0.4617994	&0.4648947	\\
9	&0.4636364	&0.4184367	&0.4664235	&0.4670655	&0.4697591	\\
10	&0.5000000	&0.4740260	&0.4761194	&0.4731737	&0.4748390	\\
11	&$\vdots$	&0.5000000           &0.4837925	&0.4804556	&0.4822382	\\
12  &	        &$\vdots$	        &0.5000000           &0.4882257	&0.4874533	\\
13	&	        &	        &$\vdots$           &0.4888889	&0.4926606	\\
14	&	        &	        &           &0.5000000	&0.5000000	\\
15	&	        &           &	        &$\vdots$	&$\vdots$	\\
16	&	        &	        &           &	&	\\
17	&	        &           &           &	&	\\\hline
$\epsilon$	&0.0564079	&0.0440134	&0.0559002	&0.0524947	&0.0531586	\\ \hline
    \end{tabular}
    \caption{Values of $\tilde{\epsilon}_h$ approximated as in Eq.\ (\ref{eq:eh}). Histories $h$ have length ranging from $0$ to the first one for which $\tilde{\epsilon}_h=\frac{1}{2}$. Results are obtained for discretization $2$ for all of the preprocessed data, i.e., $1.5$~Gb bits that are the output of discretization.}
    \label{tab:disc5}
\end{table}

\begin{table}
    \centering
    \begin{tabular}{|l|c|c|c|c|c|}
        \hline
        Data&\multicolumn{5}{c|}{Discretization type}\\
        \hline
Mb&1&2&3&4&5\\\hline
$1$&0.056412&0.056408&0.160171&0.129854&0.117910	\\
$10$	&0.044013	&0.044013	&0.349420	&0.119910	&0.117938	\\
$100$	&0.055900	&0.055900	&0.342898	&0.119000	&0.112164	\\
$1000$	&0.052490	&0.052495	&0.354050	&0.114387	&0.1072150	\\ 
$1500$  &0.053155	&0.053159	&0.347494	&0.114397	&0.106662	\\ \hline
    \end{tabular}
    \caption{Values of $\epsilon$ approximated as in Eq.\ (\ref{eq:svtest_epsilon}). The values from the rows of the table are depicted in figures \ref{fig:discrNo3} and \ref{fig:discrOnly3}. The values from the columns of the table are depicted in figures \ref{fig:bitsNo3} and \ref{fig:bitsOnly3}.}
    \label{tab:finalepsilon}
\end{table}

\clearpage

\begin{table*}
    \centering
    \begin{tabular}{|c|c|c|c|c|c|}
        \hline
        &\multicolumn{5}{c|}{Data part}\\
        \hline
h&1&2&3&4&5\\
\hline
0	&	0.000009335230267	&	0.000003310188415	&	0.000001191154359	&	0.000002107408312	&	0.000004291110569	\\
1	&	0.000025846695087	&	0.000004582753321	&	0.000005082375665	&	0.000006998426217	&	0.000023041784750	\\
2	&	0.000041590050766	&	0.000007323849964	&	0.000021254931950	&	0.000015621324682	&	0.000034600480710	\\
3	&	0.000050795833985	&	0.000018737676243	&	0.000039732694784	&	0.000032788531912	&	0.000048353111581	\\
4	&	0.000054756366839	&	0.000057293155124	&	0.000074985501337	&	0.000072820610613	&	0.000069964268718	\\
5	&	0.000109086894526	&	0.000087653258636	&	0.000112385714904	&	0.000123156489192	&	0.000104963916382	\\
6	&	0.000126993677881	&	0.000179383831999	&	0.000171806788213	&	0.000169408148117	&	0.000138753605265	\\
7	&	0.000238525787233	&	0.000257658370478	&	0.000331398847916	&	0.000292047973044	&	0.000270134018270	\\
8	&	0.000371302665140	&	0.000363343328917	&	0.000410622578234	&	0.000357706721408	&	0.000429371277418	\\
9	&	0.000577317047750	&	0.000483661467972	&	0.000686264587309	&	0.000561595178471	&	0.000578862553060	\\
10	&	0.000937648425864	&	0.000782908165541	&	0.000901548429593	&	0.000743012457801	&	0.000890928303893	\\
11	&	0.001118808877503	&	0.001276773719423	&	0.001156081983579	&	0.001156608340277	&	0.001507358668999	\\
12	&	0.001771414946252	&	0.001965556157189	&	0.001802640831625	&	0.001633033712324	&	0.002659534871718	\\
13	&	0.002892755602853	&	0.002923235112898	&	0.003182758047011	&	0.002494190795683	&	0.003326580299263	\\
14	&	0.004037554214334	&	0.003862042767375	&	0.004325725976559	&	0.004093153991520	&	0.004300613692009	\\
15	&	0.005399122493831	&	0.005512831659066	&	0.006170190787226	&	0.005881609903210	&	0.006063492063492	\\
16	&	0.008991374645915	&	0.008421620249952	&	0.008362850834697	&	0.008796979360018	&	0.008522141713307	\\
17	&	0.011502970547339	&	0.011528389857572	&	0.013996445347213	&	0.012799207439839	&	0.012590791284037	\\
18	&	0.019800698848195	&	0.017354537422693	&	0.017521069640707	&	0.019021568257561	&	0.020265054040144	\\
19	&	0.026045676161150	&	0.028502660248290	&	0.027157360406091	&	0.026028819940283	&	0.027580415966929	\\
20	&	0.039222964612720	&	0.040465351542742	&	0.039463601532567	&	0.039354838709677	&	0.039939332659252	\\
21	&	0.058823529411765	&	0.057410217501265	&	0.057916034395549	&	0.059221200648999	&	0.056767476450174	\\
22	&	0.089068825910931	&	0.085975024015370	&	0.082346609257266	&	0.081488933601610	&	0.090128755364807	\\
23	&	0.118143459915612	&	0.123721881390593	&	0.119047619047619	&	0.124713958810069	&	0.128230616302187	\\
24	&	0.184410646387833	&	0.176855895196507	&	0.175925925925926	&	0.175105485232068	&	0.196428571428571	\\
25	&	0.254385964912281	&	0.263440860215054	&	0.252577319587629	&	0.256756756756757	&	0.250000000000000	\\
26	&	0.397435897435897	&	0.371794871794872	&	0.413043478260870	&	0.369565217391304	&	0.360000000000000	\\
27	&	0.500000000000000	&	0.500000000000000	&	0.500000000000000	&	0.500000000000000	&	0.500000000000000	\\
\hline
$\epsilon$ & 0.000027701862941 & 0.000013186161706 & 0.000016819153224 & 0.000016170395862 & 0.000024421920347
\\\hline
    \end{tabular}
    \caption{Results for the data from the quantum random number generator. Values of $\tilde{\epsilon}_h$ are approximated as in Eq.\ (\ref{eq:eh}) and the final epsilons are calculated using the exponential weighted average from Eq.\ (\ref{eq:svtest_epsilon}). Histories $h$ have length ranging from $0$ to the first one for which $\tilde{\epsilon}_h=\frac{1}{2}$. Results are obtained from five portions (shown as columns) of data. Each portion consists of over 33 billion random bits.}
    \label{tab:qrng}
\end{table*}

\begin{table*}
    \centering
    \begin{tabular}{|c|c|c|c|c|c|}
        \hline
        File&\multicolumn{5}{c|}{Discretization type}\\
        \hline
&1&2&3&4&5\\\hline
ACCWF	& 0.071715892300340	& 0.071715892300340	& 0.235359194614212	& 0.047539945415787	& 0.087274616642618	\\
DISPWF	& 0.208395504327412	& 0.208395504327412	& 0.307665290116610	& 0.154972257510402	& 0.131958493734998	\\
VELWF	& 0.151712987377098	& 0.151712987377098	& 0.264583416462133	& 0.090812728222981	& 0.043024104589876	\\
ACCSWF	& 0.117156671701359	& 0.117156671701359	& 0.233832614203023	& 0.060853189218447	& 0.048078338501318	\\
DISPSWF	& 0.223122159588165	& 0.223122159588165	& 0.296299615798759	& 0.188365748213474	& 0.172095116214049	\\
VELSWF	& 0.188435869518004	& 0.188435869518004	& 0.266977575713240	& 0.133181565740910	& 0.100437780509834	\\ \hline
    \end{tabular}
    \caption{Values of epsilons for seismic events for five types of discretization, two types of time windows: ``WF'' and ``SWF'', and three signal responses: acceleration, displacement, and velocity. The ``WF'' type files consist of 284282034 data points, and the ``SWF'' type files consist of 520407646 data points.}
    \label{tab:events}
\end{table*}

\end{document}